\documentclass[11pt,oneside,letterpaper]{article}
\usepackage{amssymb}
\usepackage{amsmath}
\usepackage{amsthm}
\usepackage[dvips]{graphicx}
\graphicspath{{./fig/}}
\usepackage{setspace}
\usepackage{amsfonts}
\usepackage{fancyhdr}
\usepackage{xcolor}
\usepackage{graphicx}
\usepackage{rotating}
\usepackage{comment}
\usepackage{color}
\usepackage{cite}
\usepackage{braket}
\usepackage{overpic}

\usepackage[export]{adjustbox}
\usepackage{cancel}

\definecolor{darkgreen}{rgb}{0,0.5,0}
\definecolor{darkblue}{rgb}{0,0,0.6}
\definecolor{purple}{rgb}{0.4,.2,0.7}
\definecolor{pastelone}{rgb}{0.78,0.50,0.45}
\definecolor{pasteltwo}{rgb}{0.42,0.62,0.82}
\definecolor{pastelthree}{rgb}{0.48,0.70,0.48}

\makeatletter
\newcommand{\namedlabel}[2]{\phantomsection\def\@currentlabel{#2}\label{#1}}
\makeatother

\usepackage[normalem]{ulem}
\usepackage[shortlabels]{enumitem}
\usepackage[colorlinks=true,citecolor=darkgreen,linkcolor=black,urlcolor=purple]{hyperref}
\usepackage{sidecap}
\usepackage{pdfsync}

\allowdisplaybreaks  

\numberwithin{equation}{section}

\newtheorem{theorem}{Theorem}
\newtheorem{proposition}[theorem]{Proposition}
\newtheorem{lemma}[theorem]{Lemma}
\newtheorem{corollary}[theorem]{Corollary}
\newtheorem{conjecture}[theorem]{Conjecture}
\theoremstyle{definition}
\newtheorem{definition}[theorem]{Definition}
\newtheorem{example}[theorem]{Example}
\theoremstyle{remark}
\newtheorem{remark}[theorem]{Remark}

\DeclareMathOperator{\tr}{Tr}
\DeclareMathOperator{\id}{id}
\DeclareMathOperator{\Sym}{Sym}

\begin{document}
\onehalfspacing

\begin{center}

~
\vskip5mm

{\LARGE \bf  {Multi-entropy in random tensor networks}}

\vskip10mm

Miao Hu{${}^{1}$}, Simon Lin{${}^{2}$}, Ion Nechita{${}^{1}$},
\vskip5mm

{${\ }^{1}$}{\it Laboratoire de Physique Th\'eorique, Universit\'e de Toulouse, CNRS, UPS, France}\\
\vskip5mm
{${\ }^{2}$}{\it  New York University Abu Dhabi, Abu Dhabi, P.O. Box 129188, United Arab Emirates}
\vskip5mm

\href{mailto:miao.hu@irsamc.ups-tlse.fr}{\texttt{miao.hu@irsamc.ups-tlse.fr}},
\href{mailto:simonlin@nyu.edu}{\texttt{simonlin@nyu.edu}},
\href{mailto:ion.nechita@univ-tlse3.fr}{\texttt{ion.nechita@univ-tlse3.fr}}

\end{center}

\vspace{4mm}

\begin{abstract}
\noindent
We study the evaluation of R\'enyi multi-entropies $S^{(q)}_n$ in Random Tensor Network (RTN) states in the large bond-dimension limit. For the case of R\'enyi index $n=2$ and arbitrary number of parties $q$, we prove that that multi-entropies are determined by minimal multiway cuts through the network. When the minimal multiway cut is degenerate, we characterize the full minimizer set via compatible families of minimal cuts and give a criterion for all minimizers to come from ordinary cut partitions. For $n=2$, this gives a natural generalization of the minimal cut description of bipartite entanglement to multipartite systems with arbitrarily many parties. For the case of integer $n>2$, we show that the minimal multiway cut conjecture is in general \emph{not true} by providing explicit counter examples for both the single random tensor and for the network built from isometric tilings. We discuss the implication for our results on the multipartite entanglement structures in RTN and holography.
\vspace{1cm}
\begin{equation*}
        \includegraphics[width=0.4\linewidth]{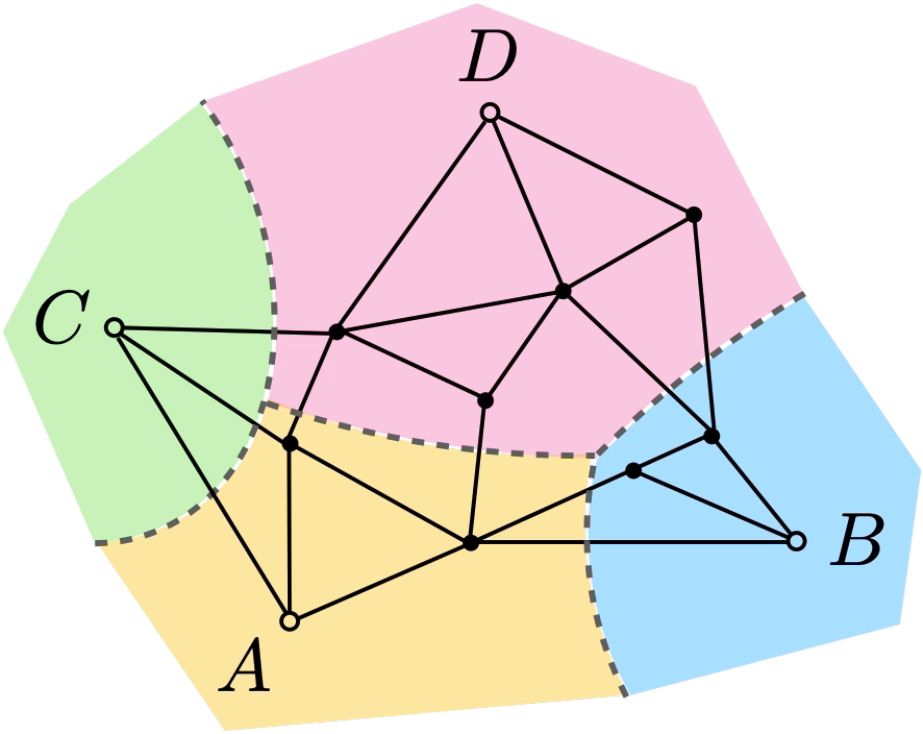}
\end{equation*}
 \end{abstract}

\pagebreak
\pagestyle{plain}

\setcounter{tocdepth}{2}
{}
\vfill
\tableofcontents

\newpage

\section{Introduction and summary of results}

It is now well understood that quantum entanglement plays a pivotal role in the emergence of spacetime in holography and quantum gravity. One of the most prominent example is the Ryu--Takayanagi (RT) formula \cite{Ryu:2006bv,Ryu:2006ef}, relating the area of minimal surfaces in the bulk to entanglement entropy of the boundary theory. The RT formula marked the beginning of the ``geometry from entanglement'' paradigm, where one expresses the entanglement data of the boundary CFT in terms of various  geometric objects in the bulk\cite{Hubeny:2007xt,Dong:2016fnf,Takayanagi:2017knl,Umemoto:2018jpc,Bao:2019zqc,Dutta:2019gen,Dong:2021clv,Nakata:2020luh}.

Random tensor networks (RTNs) \cite{Collins:2013ubu,Hayden:2016cfa} model this behavior by contracting randomly sampled tensors according to a graph $G=(E,V)$ that plays the role of the underlying bulk geometry. This defines a state on a boundary Hilbert space associated with a selected set of boundary vertices $\partial\subset V$. The RT formula then arises as the area of a minimal cut across $G$ that divides the boundary into two disconnected components. See Figure \ref{fig:TN_cuts_intro}a. Although RTNs have their limitations when it comes to time dependent or covariant physics \cite{Dolev:2021ofc,Faist:2019ahr,Engelhardt:2023bpv}\footnote{See also \cite{Apel:2021tnn,Kohler:2018kqk,Balasubramanian:2025rcr} for some attempts to circumvent this issue.}, it remains a good toy model for demonstrating the ``geometry from entanglement'' paradigm and is believed to accurately capture the physics of fixed--area states \cite{Akers:2018fow,Dong:2018seb} in holography.

While most of our understanding of the holographic principle in this area has insofar come from the entanglement entropy, the entanglement entropy by itself is a rather coarse measure of entanglement in the sense that it only captures the amount of \emph{bipartite} entanglement between two different subsystems. Moreover, in recent years, there has been increasingly compelling evidence suggesting that multi-partite entanglement is more common in holographic states than previously believed \cite{Akers:2019gcv,Harper:2021uuq,Li:2025nxv}, and that it may play important roles in the holographic dictionary \cite{Iizuka:2025bcc}.

\begin{figure}[th]
  \centering
  \includegraphics[width=.8\linewidth]{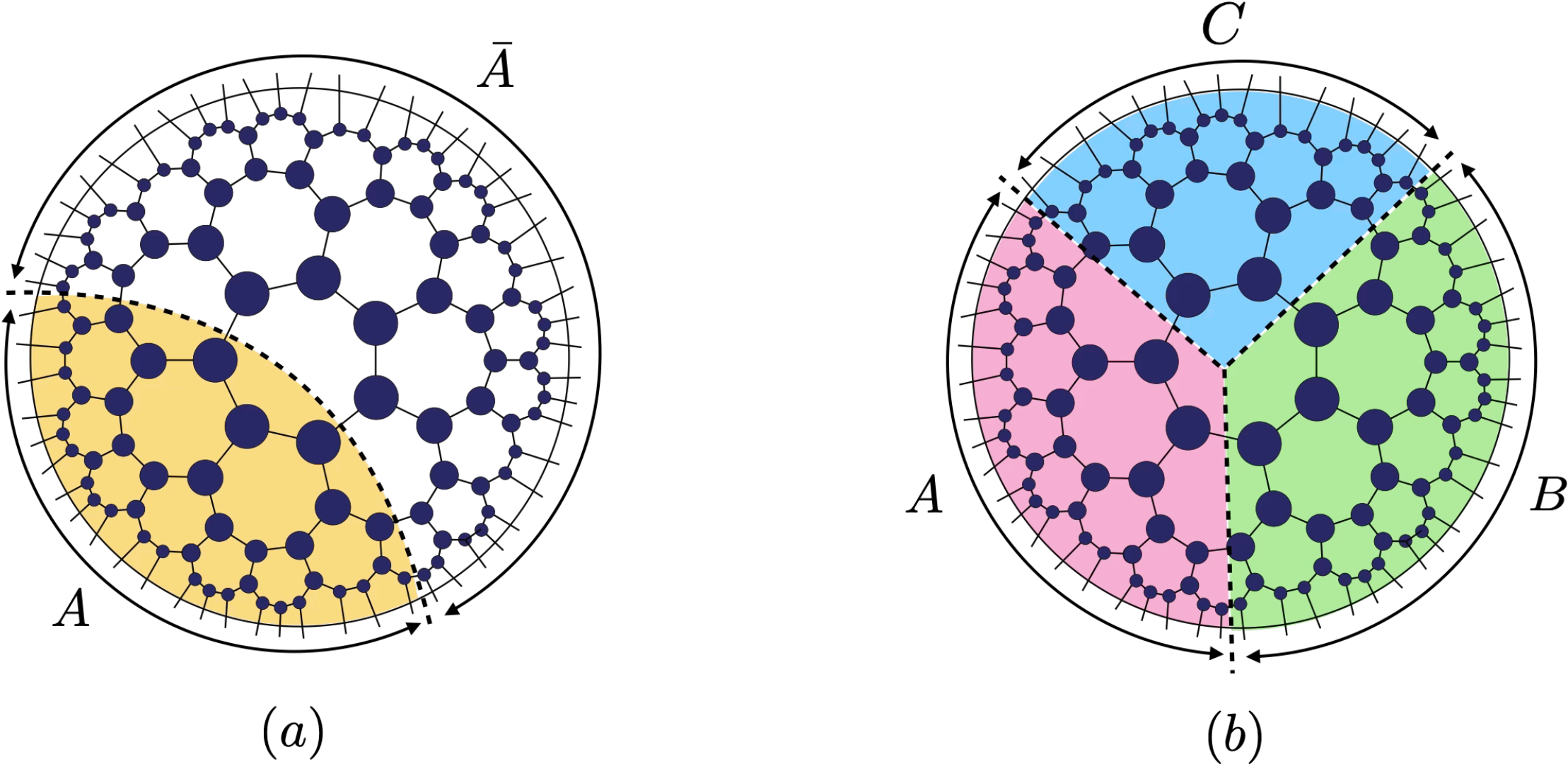}
  \caption{We tile the hyperbolic disk isometrically using a graph. Each dark blue node represents a random tensor and the connecting edges indicate tensor contractions. The random tensor network defines a state on the Hilbert space supported on the boundary, represented here by dangling legs. (a) The geometric dual of the entanglement entropy: The entropy of region $A$ is proportional to the area of the minimal surface (black dashed) that divides $A$ and its complement $\bar{A}$. (b) The proposed geometric dual for the multi-entropy. Here we take $q=3$ as an example. The tripartite multi-entropy between region $A:B:C$ is proportional to the area of the minimal multiway cut (also black dashed) that divides the three boundary regions $A,B$ and $C$.}
  \label{fig:TN_cuts_intro}
\end{figure}

To better study multi-partite entanglement in holography, it is important that we look at suitable quantifiers of multi-partite entanglement. A natural candidate to consider is the \emph{multi-entropy} $S^{(q)}$, introduced in Ref. \cite{Gadde:2022cqi} as a natural generalization of the entanglement entropy to $q$-partite systems. It has been shown that $S^{(q)}$ is sensitive to multi-partite entanglement. Furthermore, by taking suitable linear combinations of different multi-entropies, one can construct measures that act as signals for genuinely multi-partite entangled states, i.e. quantities that vanish if the state can be factorized as products of lower partite entangled states \cite{Iizuka:2025ioc,Iizuka:2025caq,Gadde:2026msg}.\footnote{See \cite{Harper:2025uui,Berthiere:2025toi,Yuan:2025dgx,Ju:2025eyn,Anegawa:2025prn,Iizuka:2025elr,Iizuka:2025pqq,Chen:2026xtx,DelZotto:2026fpw,Iizuka:2026ahd,Iizuka:2026qqg} for a selection of recent developments on the multi-entropy.}

One immediate question is whether multi-entropy satisfies some generalized version of RT formula for holographic states. The standard argument for the RT formula, \'a la Lewkowycz--Maldacena \cite{Lewkowycz:2013nqa}, involves the replica trick for the R\'enyi entropy $S_n$ and a suitable analytic continuation $n\to 1$ of the corresponding bulk geometry. Applying the same procedure to the R\'enyi version of the multi-entropy, one finds that $S^{(q)}$ admits a natural bulk geometric dual as \emph{minimal multiway cuts} \cite{Gadde:2022cqi}: The minimal area of the ``soap film'' that is anchored on the interface between boundary regions. See Figure \ref{fig:TN_cuts_intro}b for an illustration.

One critical assumption in the derivation of Ref. \cite{Lewkowycz:2013nqa} is that the dominant bulk saddle for the R\'enyi entropies must be \emph{replica symmetric}. While this is indeed true for the case of the (R\'enyi) entanglement entropy, its validity for (R\'enyi) multi-entropy is doubtful. Indeed, Ref. \cite{Penington:2022dhr} gave a counter example for $(q,n)=(3,3)$ where the dominant saddle is given by a handle-body and not replica symmetric. For arbitrary $q$, Ref. \cite{Gadde:2024taa} demonstrated the existence of a replica-symmetric saddle for $n=2$ but not for higher $n$ in general, thus putting a big question mark on the analytic continuation required to obtain the generalized RT formula.

The situation is slightly more optimistic if one instead considers RTN states (or fixed-area states), since the bulk geometry for these states are more or less ``fixed'' so one can exclude bulk handle-body solutions. As such, it has been conjectured that the multi-entropy in RTN states are given by the area of minimal multiway cuts \cite{Iizuka:2024pzm,Carrozza:2026qcf}. In particular, the case $(q,n)=(3,2)$ has been proven in Ref. \cite{Penington:2022dhr} for RTNs with a unique minimum multiway cut.\footnote{The $(q,n)=(3,2)$ R\'enyi multi-entropy is also closely related to the $(m,n)=(2,2)$ R\'enyi reflected entropy \cite{Akers:2024pgq} and the $n=2$ computable cross norm \cite{Milekhin:2022zsy}, where similar results has been obtained.}
However, fully determining the multi-entropy still requires summing over all possible replica permutations, and it is not clear that one can recover the multiway cut at the end of the day. There is limited evidence supporting the multiway cut conjecture. For example, Ref. \cite{Iizuka:2024pzm} performed exact searches for low $(q,n)$ and found that the conjecture is valid for the cases examined.\footnote{See also \cite{Akella:2026bci} for related discussion on the effect of replica-symmetry breaking of multi-entropy on holographic RTNs.} A full systematic treatment remains absent for general cases.

In this paper, we solve the problem of determining R\'enyi multi-entropies in various RTN states. Working in the limit of large bond dimensions, we establish the following results:
\begin{itemize}
\item For the R\'enyi index $n=2$ and arbitrary number of parties $q$, we prove the multiway cut conjecture for multi-entropy. More specifically, we show that the $q$-partite R\'enyi multi-entropy $S^{(q)}_2$ is proportional to the area of minimal multiway cuts across the network.

\item For the R\'enyi index $n=2$ and arbitrary number of parties $q$, we show that there exists a unique minimal saddle given that the minimal mutliway cut of the network is unique. For the case where the minimal cut is not unique, we give criteria where the minimal saddles must satisfy. We explicitly count the number of minimal saddles for some simple classes of networks.

\item For higher integer R\'enyi index $n>2$, we show that the multiway cut conjecture is \emph{false} by giving an explicit counter example for each $(q,n)$ pair. Our counter example consists of a $q$-partite single random tensor, where we explicitly demonstrate that there exists a more dominant saddle than the one predicted by minimal multiway cut.\footnote{With the sole exception of $(q,n)=(3,3)$, where our counter example gives the same estimate as the replica-symmetric one.} We argue that the existence of the counter example in the single random tensor setting implies the existence of a whole class of replica-symmetry-broken saddles for RTNs defined on isometric tilings of a metric space for large enough $n$.

\end{itemize}

This paper is organized as follows:
In Section \ref{sec:prelim} we present a brief review of materials necessary for proving our main results.
In Section \ref{sec:n=2} we give a proof for the multiway cut conjecture for $n=2$ R\'enyi multi-entropies. We first present the proof with a for  a single random tensor before moving on to deal with RTN states defined on arbitrary graphs.
In Section \ref{sec:minimizer} we further develop our results for the $n=2$ case, focusing on the networks with degenerate minimizers. We provide a list of criteria a minimizer must satisfy. This allows us to explicitly count the number of minimizers for some simple class of RTNs.
In Section \ref{sec:n>2} we deal with the higher $n>2$ R\'enyi multi-entropies. We present a counter-example to the multiway cut conjecture in a simple one tensor network and discuss its implication on more general RTNs.
We end this paper with discussion and future directions in Section \ref{sec:discussion}.

\section{Preliminaries}
\label{sec:prelim}
We begin with a brief review of materials necessary for proving our main results.

\subsection{Multi-entropy}
Consider a quantum state $\ket{\psi}$ in a finite dimensional Hilbert space $\mathcal{H}$. Suppose that $\mathcal{H}$ can be written as the tensor product of $q$ different subspaces
\begin{equation}
  \mathcal{H} = \mathcal{H}_1 \otimes \mathcal{H}_2 \otimes \cdots \otimes \mathcal{H}_q.
\end{equation}
We can decompose $\ket{\psi}$ using
\begin{equation}
  \label{eq:psi_decomp}
  \ket{\psi} = \sum_{i_1}\sum_{i_2}\cdots \sum_{i_q} \psi_{i_1i_2\cdots i_q}\ket{e^{i_1}_1} \ket{e^{i_2}_2} \cdots \ket{e^{i_q}_q},
\end{equation}
where $\ket{e^{i}_a}$ represents a set of basis vectors in $\mathcal{H}_a$.
The coefficient $\psi_{i_1 \cdots i_q}$ can be viewed as a rank-$q$ tensor that transforms contravariantly under the action of local unitaries (LU). In contrast, the complex conjugate $\bar{\psi}^{j_1\cdots j_q}$ are the coefficient of the dual vector $\bra{\psi}$, which transforms covariantly under LU. One can build quantities known as \emph{multi-invariants} \cite{Gadde:2024taa} that is LU-invariant by fully contracting an equal number of $\psi$ and $\bar{\psi}$ along their common indices.
Let $N$ be the number of $\psi$ copies (or equivalently, the number of $\bar{\psi}$ copies). We can construct any contraction pattern by specifying how each individual subspaces are contracted, using an element in the permutation group $\Sym_N$ (often called the \emph{replica symmetry group}). There are $q$ subsystems, and thus any multi-invariant is uniquely specified by $q$ elements $g_1,\cdots g_q\in \Sym_N$ (often called the \emph{twist operators}), as shown in the following definition:
\begin{definition}[multi-invariants]
  \label{def:multi-invariants}
  Let $\ket{\psi}$ be a $q$-partite pure state with decomposition of \eqref{eq:psi_decomp} and $g_1,\cdots g_q\in \Sym_N$ be a list of permutations in the replica symmetry group. The multi-invariant $\mathcal{Z}$ associated to the list $\{g_a\}$ is defined as
  \begin{equation}
    \mathcal{Z}(g_1,\cdots g_q) = \Big(\psi_{i^{(1)}_1\cdots i^{(1)}_q}\cdots\psi_{i^{(N)}_1\cdots i^{(N)}_q}\Big) \Big(\bar{\psi}^{i^{(g_1\cdot1)}_1\cdots i^{(g_q\cdot1)}_q}\cdots\bar{\psi}^{i^{(g_1\cdot N)}_1\cdots i^{(g_q\cdot N)}_q}\Big),
  \end{equation}
  where repeated indices are summed over.
  Or equivalently, let $\rho_{i_1\cdots i_q}^{j_1\cdots j_q}=\psi_{i_1\cdots i_q}\bar{\psi}^{j_1\cdots j_q}$ be the coefficients for the density matrix $\rho=\ket{\psi}\bra{\psi}$, we have
  \begin{equation}
    \mathcal{Z}(g_1,\cdots g_q) = \rho_{i^{(1)}_1\cdots i^{(1)}_q}^{i^{(g_1\cdot 1)}_1\cdots i^{(g_q\cdot1 )}_q}\cdots\rho_{i^{(N)}_1\cdots i^{(N)}_q}^{i^{(g_1\cdot N)}_1\cdots i^{(g_q\cdot N )}_q}.
  \end{equation}
\end{definition}
Definition~\ref{def:multi-invariants} is a bit unwieldy to work with, so let us look at some examples.
\begin{example}[R\'enyi entropy]
Let $\ket{\psi}\in \mathcal{H}_A\otimes \mathcal{H}_B$ be a bipartite state. The bipartite R\'enyi entropy is $S_n(A)=\frac{1}{1-n}\ln (\tr \rho^n_A)$ where $\rho_A=\tr_B\ket{\psi}\bra{\psi}$ is the reduced density matrix. The moment invariant $\tr \rho^n_A$ is an LU-invariant with twist operators
\begin{equation}
  g_A = \tau_n = (12\cdots n), \quad g_B = \id.
\end{equation}
The replica group for the moment is $\Sym_n$.
\end{example}
\begin{example}[Entanglement negativity \cite{Vidal:2002zz}]
  \label{ex:negativity}
  Let $\ket{\psi}\in \mathcal{H}_A \otimes \mathcal{H}_B \otimes \mathcal{H}_C$ be a tripartite state and $\rho_{AB}=\tr_C \ket{\psi}\bra{\psi}$ be the reduced density matrix. The R\'enyi negativity is defined as (the logarithm of) the moments of the partial transpose $\rho^{T_B}_{AB}$. The twist operators for the R\'enyi negativity is
  \begin{equation}
    g_A = \tau_n = (12\cdots n),\quad g_B = \tau^{-1}_n = (n\cdots 21), \quad g_C=\id.
  \end{equation}
  The replica group for the partial transpose moment is also $\Sym_n$.
\end{example}
\begin{example}[Reflected entropy \cite{Dutta:2019gen}]
  \label{ex:S_R}
  Let $\ket{\psi}$ and $\rho_{AB}$ be the same as the previous example. For integer $n$ and even integer $m$, the R\'enyi reflected entropy is the R\'enyi entropy $S_n(AA^*)$ for $\ket{\rho_{AB}^{m/2}}$, viewed as a vector in the doubled Hilbert space $\mathcal{H}_{AB}\otimes H_{A^*B^*}$ via the usual state-operator correspondence. The twist operators for the R\'enyi reflected entropy is
  \begin{align}
  \begin{split}
    g_A &= (1\cdots m)(m+1\cdots2m)\cdots (nm-m+1\cdots mn),\\
    g_B &= (\tfrac{m}{2}+1\cdots \tfrac{3m}{2})(\tfrac{3m}{2}+1\cdots \tfrac{5m}{2})\cdots(nm-\tfrac{m}{2}+1\cdots \tfrac{m}{2}), \\ g_C&=\id.
  \end{split}
  \end{align}
  The replica group for R\'enyi reflected entropy is $\Sym_{mn}$.
\end{example}

Before proceeding, let us introduce a diagrammatic notation for better visualization of the multi-invariants. Notice that $\mathcal{Z}$ is invariant when we multiply all the twist operators by an arbitrary permutation $g_a\to hg_a$. We can use this freedom to fix one of them (say $g_q$) to be the identity permutation. In terms of the density matrix $\rho=\ket{\psi}\bra{\psi}$, this corresponds to taking the partial trace over the Hilbert space $\mathcal{H}_q$.
We denote the reduced density matrix $\rho_{i_1\cdots i_{q-1}}^{j_1\cdots j_{q-1}}$ after tracing out $\mathcal{H}_q$ as a vertex with $2(q-1)$ legs. Each leg corresponds to one index $i_a$ or $j_a$. The multi-variant $\mathcal{Z}$ is constructed by first laying out $N$ copies of vertices and then contracting the legs according to the order given by the twist operators. See Fig.~\ref{fig:contract} for an illustration.
\begin{figure}[t]
  \centering
  \includegraphics[width=.75\linewidth]{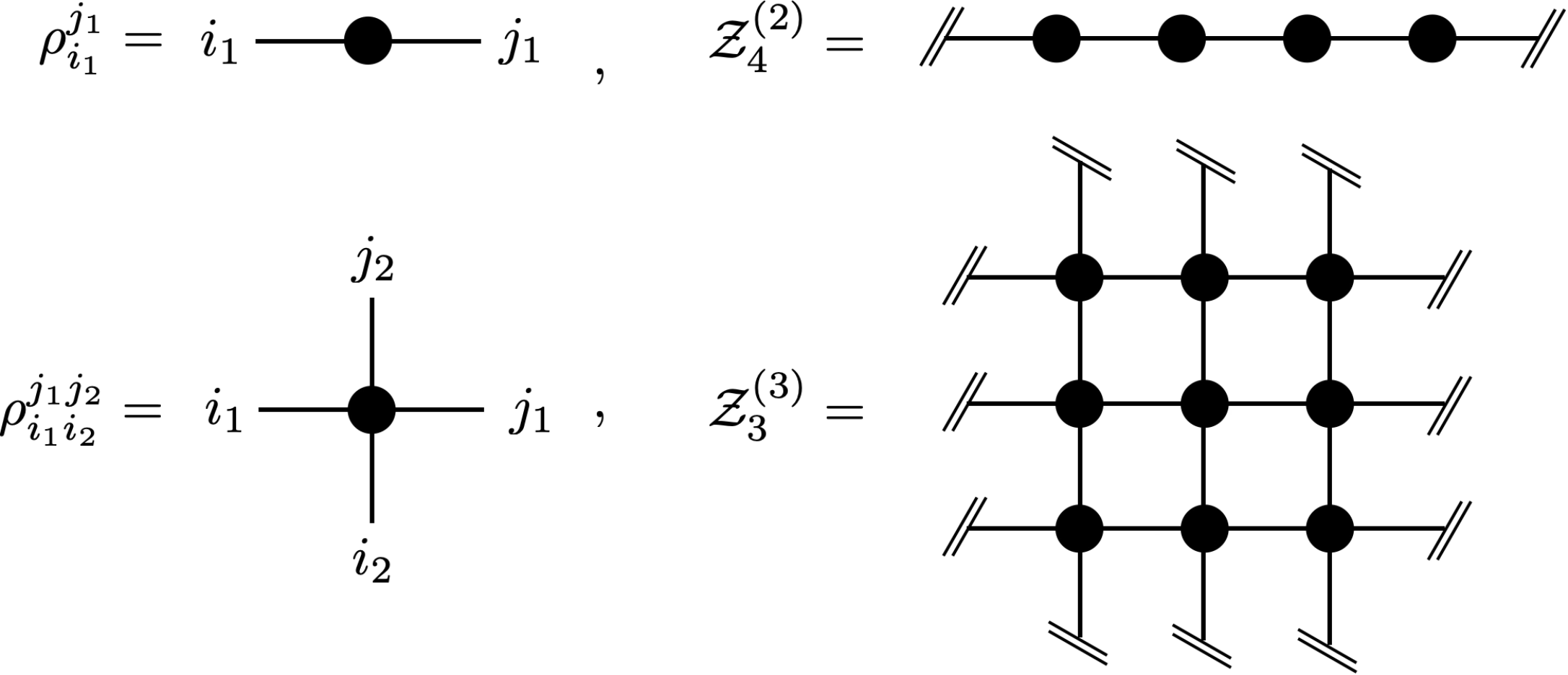}
  \caption{The diagrammatic notation for the reduced density matrix $\rho^{j_1\cdots j_q}_{i_1\cdots i_q}$, and the contraction pattern of multi-invariant $\mathcal{Z}^{(q)}_n$ used in defining the multi-entropy.}
  \label{fig:contract}
\end{figure}

The R\'enyi multi-entropy is a natural generalization for the bipartite R\'enyi entropy to higher partite systems \cite{Gadde:2022cqi}. The twist operators acts as translations on a periodic hypercube lattice of $(q-1)$ dimensions.
\begin{definition}[R\'enyi multi-entropy]
\label{def:mutli_ent}
  Let
  \begin{equation}
     X=\{0,\cdots,n-1\}^{q-1}=\mathbb{Z}_n^{q-1},\quad |X|=n^{q-1}
  \end{equation}
  be the set of $(q-1)$-dimensional integer lattice points on a hypercube of length $n$. The \emph{R\'enyi multi-entropy} $S^{(q)}_n$ is defined in terms of the multi-invariant $\mathcal{Z}(g_1,\cdots,g_q)$ as
  \begin{equation}
    S^{(q)}_n(R_1:\cdots :R_q) := \frac{1}{1-n}\frac{1}{n^{q-2}} \ln \mathcal{Z}(g_1,\cdots,g_q),
  \end{equation}
  where the twist operators $g_i$ act on $X$ by translation along the $i$-th axis, i.e.
  \begin{align}
    &g_i \cdot (x_1, \cdots ,x_i,\cdots,x_{q-1}) = (x_1, \cdots ,x_i+1,\cdots,x_{q-1}),\quad  i<q, \\
    &g_q = \id,
  \end{align}
   where the arithmetic operations on $\mathbb{Z}_n$ are understood to be up to modulo $n$.
 \end{definition}
 The replica symmetry group for the R\'enyi multi-entropy is $\Sym(X)=\Sym_{n^{q-1}}$. The contraction patterns of $S^{(q)}_n$ for some low $(q,n)$ are illustrated in Fig.~\ref{fig:contract}. For $q=2$, it reduces back to the usual bipartite R\'enyi entropies $S_n$.

R\'enyi multi-entropies enjoy the following properties \cite{Gadde:2022cqi,Gadde:2023zni}:
 \begin{itemize}
 \item (LU invariance). $S^{(q)}_n(R_1:\cdots:R_q)$ is invariant under LU operations.
 \item (Symmetry). $S^{(q)}_n(R_1:\cdots:R_q)$ is invariant under permutations of the subsystems $R_i\to R_{\sigma\cdot i}$ for any $\sigma\in S_q$.
 \item (Additivity). $S^{(q)}_n(R_1:\cdots:R_q)$ is additive under tensor products: Let $\ket{\psi}\in H_{A_1\cdots A_q}$ and $\ket{\phi}\in H_{B_1\cdots B_q}$ be two $q$-partite states, then
   \begin{equation}
     S^{(q)}_n(A_1B_1:\cdots:A_qB_q)_{\ket{\psi}\otimes \ket{\phi}} = S^{(q)}_n(A_1:\cdots:A_q)_{\ket{\psi}} + S^{(q)}_n(B_1:\cdots:B_q)_{\ket{\phi}}.
   \end{equation}
 \end{itemize}
 Just as bipartite R\'enyi entropy can be used to detect bipartite entanglement, R\'enyi multi-entropies can serve as signals for multi-partite entanglement. In particular, one can form linear combinations of $S^{(q)}_n$ called \emph{Genuine multi-entropy} that precisely detects genuine $q$-partite entanglement, which are states that cannot be written as tensor products of lower-partite entangled states. See Ref. \cite{Iizuka:2025ioc,Iizuka:2025caq} for more details on how to construct Genuine multi-entropies.

\subsection{Random tensor networks}
\begin{definition}[Random tensor network states \cite{Hayden:2016cfa}]
  Let $G=(V,E)$ be a graph with edge weight $\chi:E\to \mathbb{N}$. For each vertex $v\in V$ we assign a Hilbert space
  \begin{equation}
    \mathcal{H}^v = \bigotimes_{e\in E(v)} \mathcal{H}^{(v)}_e,\quad \mathcal{H}^{(v)}_e = \mathbb{C}^{\chi(e)},
  \end{equation}
  where $E(v)=\{(u,v):u\in N(v)\}$ is the set of edges emanating from $v$.
  We mark a subset $\partial\subset V$. We refer to $\partial$ as the set of \emph{boundary (terminal) vertices} and $V\setminus\partial$ as the set of \emph{bulk (interior) vertices}.
  The random tensor network assigns $G$ a pure state $\ket{\psi}_\partial$ on its boundary vertex Hilbert space $\mathcal{H}_\partial=\bigotimes_{\mathit{v}\in \partial} \mathcal{H}_v$ in the following manner: For each bulk vertex we sample a random vector $T(v)\in \mathcal{H}_v$ according to the Haar measure, and construct
  \begin{equation}
    \ket{\psi}_\partial \propto \left(\bigotimes_{v\in V\setminus \partial} \bra{T(v)} \right) \left(\bigotimes_{e\in E} \ket{\Psi_e}\right),
  \end{equation}
  where $\ket{\Psi_e}\in \mathcal{H}^{(u)}_e\otimes \mathcal{H}^{(v)}_e$ is a maximally entangled state between $\mathcal{H}^{(u)}_e$ and $\mathcal{H}^{(v)}_e$ for an edge $e=(u,v)$.
\end{definition}
In this paper we will work in the limit of \emph{large bond dimensions}, i.e. we take $\chi(e)\to \infty$ uniformly across all $e\in E$. For convenience we define the rescaled weight
\begin{equation}
  w(e) = \frac{\ln \chi(e)}{\ln \chi}
\end{equation}
which we hold finite as we take $\chi\to \infty$.

We are interested in the quantification of multi-partite entanglement in random tensor network states. To do this we divide the boundary vertices into $q$ disjoint parties $\partial = R_1 \sqcup R_2 \sqcup \cdots \sqcup R_q$, and denote the resulting $q$-partite state $\ket{\psi_{12\cdots q}}$. We want to calculate the R\'enyi multi-entropies for this state. A remarkable result in Haar-random tensor networks is that the evaluation of multi-invariants for $\psi$ can be mapped to an equivalent problem of evaluating the partition function on a associated ``Ising model'' on $G$ \cite{Hayden:2016cfa}.
\begin{align}
\label{eq:Sq_RTN}
  S^{(q)}_n(R_1:\cdots:R_q) &= \frac{1}{1-n}\frac{1}{n^{q-2}}\ln \frac{Z^{(q)}_n}{(Z^{(q)}_1)^{n^{q-1}}},\\
  Z^{(q)}_n &= \sum_{g} \exp\left(-\sum_{e=(u,v)\in E} d(g(u),g(v))\ln \chi(e)\right),
\end{align}
where the map $g:V\to \Sym(X)$ assigns an element of the replica symmetry group $\Sym(X)$ to each vertex. The ``interaction'' for this Ising model is given by the Cayley distance on $\Sym(X)$:
\begin{equation}
  d(g,h) =  |X| - \#(gh^{-1}),
\end{equation}
where $\#(g)$ counts the number of cycles of $g$ (including unit cycles).
The Cayley distance measures the minimal number of transpositions one needs to apply to change $g$ to $h$.\footnote{$d(g,h)$ is a distance function on $\Sym_N$, meaning that it satisfies: 1. $d(g,h)=d(h,g)$ (symmetry), 2. $d(g,h)\ge0$ and is zero iff $g=h$ (semi-positivity), and 3. $d(g,h)+d(h,k)\ge d(g,k)$ (triangle inequality).} The boundary condition for this Ising model is fixed by the action of corresponding twist operators on the subregions. In our case we have
\begin{equation}
  g(v) = g_i, \quad v\in R_i
\end{equation}
with $g_i$ given in Definition \ref{def:mutli_ent}.

In the large bond-dimension limit $\chi\to \infty$, one can obtain a good approximation for the partition function by saddle point approximation, i.e. finding the configuration $g$ that maximizes the exponential and dropping all the other terms:
\begin{align}
\label{eq:Z(g)}
  \log Z^{(q)}_n &\approx \max_g\left(-\sum_{e=(u,v)\in E} d(g(u),g(v))\ln \chi(e)\right) \\
  &= -(\ln\chi) \min_g\left(F(g)\right),
\end{align}
where $g$ is now the configuration that minimizes the ``free energy''
\begin{equation}
\label{eq:F(g)}
  F(g) = \sum_{e=\{x,y\}\in E} d(g(x),g(y)) w(e).
\end{equation}
This expression is the starting point of our main results in this paper. In Sec.~\ref{sec:n=2}, we will prove that for $n=2$ and arbitrary $q$, the minimizer of \eqref{eq:F(g)} is given by the solution of the multiway cut problem on the graph $G$, which in turn allows one to determine the multi-entropy of the RTN state $\ket{\psi}$ through \eqref{eq:Z(g)} and \eqref{eq:Sq_RTN}. We review multiway cut problem the following subsection.

\subsection{Multiway cuts and minimal surfaces}
Given any network $(G,w)$ and set of boundary permutations $g_i$, finding the configuration $g(v)$ that minimizes the free energy \eqref{eq:F(g)} is a highly non-trivial task.
There is however a simple solution that can be easily constructed: one simply divides $G$ into $q$ different ``domains'' $\Gamma_i$, each containing the corresponding boundary region $R_i$. One then assigns $g(v)=g_i$ for all the vertices $v$ inside the domain $\Gamma_i$. The contribution to the free energy for this configuration thus solely comes from the interface between these regions. We then pick the domain configurations $\{\Gamma_i\}$ that minimizes the total interface area. This produces a candidate $g$ for our problem.
The problem of finding the domains with minimal area is known as the \emph{multiway cut problem}, which we now define:

\begin{definition}[multiway partitions and cuts]
  Let $G=(V,E)$ be a graph with edge weight $w:E\to \mathbb{R}_+$, and $R_1,\cdots,R_q\subset V$ be a set of boundary vertices.
  A \emph{multiway partition} is a partition $V=\Gamma_1\sqcup\cdots\sqcup \Gamma_q$ such that $R_i\subseteq\Gamma_i$ for all $i=1,\cdots,q$.
  A \emph{multiway cut (or cut set)}, denoted $\mu(\{\Gamma_k\})$, is the set of edges that lies on the boundaries between any two subsets, i.e.
  \begin{equation}
    \mu (\{\Gamma_k\}) = \{ (u,v)\in E : u\in \Gamma_i \text{ and } v\in \Gamma_j,\quad i\ne j  \}.
  \end{equation}
  The \emph{area} of a multiway cut, denoted $\mathcal{A}(\{\Gamma_k\})$, is the weighted sum of $\mu(\{\Gamma_k\})$ over the edge weight $w$:
  \begin{equation}
    \mathcal{A}(\{\Gamma_k\}) = \sum_{e\in \mu (\{\Gamma_k\})} w(e).
  \end{equation}
\end{definition}
\begin{definition}[minimal multiway cut]
  The \emph{minimal multiway cut} is a multiway cut that minimizes the cut area. Stated differently, the minimal multiway cut is the minimal number of edges (weighted by $w$) one has to remove in order to separate each pair of boundary regions $R_k$.
  We denote the the area of the minimal cut as
  \begin{equation}
    \mathcal{A}(R_1:\cdots:R_q) = \min_{\Gamma_1\sqcup\cdots\sqcup\Gamma_q=V} \mathcal{A}(\{\Gamma_k\}).
  \end{equation}
\end{definition}
See Figure \ref{fig:TN_cuts} for an example.

\begin{figure}[t]
    \centering
    \includegraphics[width=0.35\linewidth]{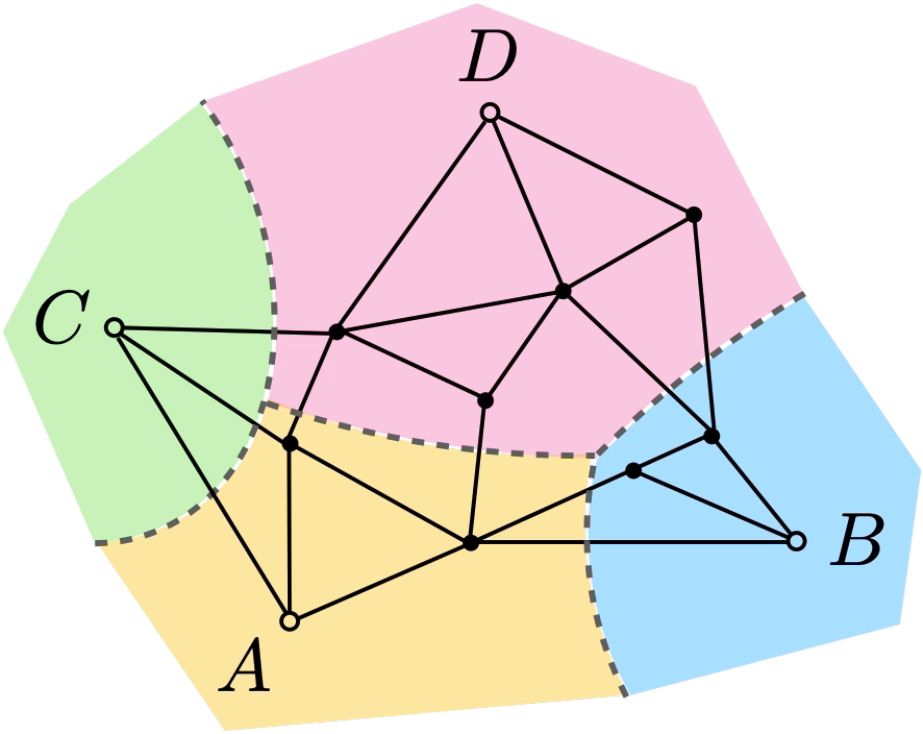}
    \caption{An example of a quadripartite minimal cut. We represent the internal vertices of $G$ with solid nodes and the terminal vertices with hollow nodes. The minimal multiway partition are shown using colored patches and the multiway cut are shown as the dashed line.}
    \label{fig:TN_cuts}
\end{figure}

Obviously, any such configuration constructed from a minimal multiway partition is not guaranteed to be a true global minimizer for $F(g)$. A remarkable result of RTN is that in the bipartite setting $\partial = A\sqcup B$, the minimal cut \emph{is} a true minimizer: the configuration
\begin{equation}
    g(v)=g_i, \quad v\in \Gamma_i, \quad i=A,B
\end{equation}
can be shown to minimize the free energy $F$.\footnote{Note that there can be other minimizers when the minimal partition is not unique.} The proof of this result involves ideas from non-crossing permutations and flow networks which we will not go into detail here. Interested readers can refer to Ref. \cite{Collins:2013ubu} for details. See also Ref. \cite{cheng2024random,Fitter:2024nxn,Hu:2025geh} for recent developments.
The minimal area prescription for bipartite quantities is what makes RTN states special (in the sense that it is always ``maximally entangled'' across any bipartite regions \cite{Swingle:2009bg}) and thus a simple toy model for studying holography \cite{Hayden:2016cfa}.

For higher parties the situation is much more complicated. Even in tripartite settings $\partial = A\sqcup B\sqcup C$, the minimal configuration can be drastically different for different measures of interest. For example, if there exists a permutation element $\gamma\in \Sym_N$ that lies on the common geodesics\footnote{A \emph{geodesic} $[g_1,g_2]$ is the set of permutation elements that saturate the triangle inequality, i.e. $\{h\in \Sym_N:d(g_1,h)+d(h,g_2)=d(g_1,g_2)\}$.}  between any pairs of $\{g_A,g_B,g_C\}$:
\begin{equation}
  \gamma \in [g_A,g_B]\cap [g_B,g_C] \cap [g_A,g_C],
\end{equation}
then it is possible to show that the minimal configuration for $F(g)$ is given by a collection of minimal bipartite cuts $r_i$ and a remainder region \cite{Shapourian:2020mkc,Dong:2021clv}:
\begin{equation}
\label{eq:negativity_sol}
  g(v)=
  \begin{cases}
    g_i, \quad &v\in r_i, \quad i=A,B,C, \\
    \gamma, \quad &v\in V\setminus r_A\setminus r_B\setminus r_C.
  \end{cases}
\end{equation}
This is exactly the case for entanglement negativity (Example \ref{ex:negativity}): the permutation $\gamma=(12)(34)\cdots(n-1 ~n)\in \Sym_n$ satisfies the above criteria. The minimal value of $F(g)$ for these kind of solutions are given by a different minimization problem on the network called \emph{multicommodity flow problem} \cite{Dong:2021clv}.

Even if such $\gamma$ does not exist, there is still no guarantee that the minimal $g$ only features the terminal permutations $g_i$, and therefore the solution is not given by a simple multiway cut. This happens to be the case for the (R\'enyi) reflected entropy (Example \ref{ex:S_R}). There exists a special element $x=(1\cdots\frac{m}{2})(\frac{m}{2}+1\cdots m)\cdots (nm-\frac{m}{2}+1\cdots nm)\in \Sym_{mn}$ such that it lies on the intersection of the geodesic $[g_A,g_C]$ and $[g_B,g_C]$ while being farthest from $g_C$. The minimal configuration has been shown to be given by the following assignment \cite{Akers:2021pvd,Akers:2022zxr,Akers:2024pgq}:
\begin{equation}
  g(v) =
  \begin{cases}
    g_i, \quad &v\in \Gamma^t_i, \quad i=A,B \\
    g_C=\id, \quad &v\in r_C \\
    x, \quad &v\in \Gamma^t_C \setminus r_C,
  \end{cases}
\end{equation}
where $r_C$ is a minimal bipartite cut of $C$ and $V=\Gamma^t_A\sqcup\Gamma^t_B\sqcup\Gamma^t_C$ is a minimal multiway partition with tension $t=(t_{AB},t_{BC},t_{AC})=(2(n-1),n,n)$.\footnote{Here a multiway partition with tension $t$ means when calculating the area of the cut, we instead multiply the edges between different partitions with different weighing factors governed by $t$. More specifically, the area is defined by
  \begin{equation}
    \mathcal{A}_t(\Gamma)= t_{AB}\sum_{e\in\mu_{AB}} w(e) + t_{BC}\sum_{e\in\mu_{BC}} w(e) + t_{AC}\sum_{e\in\mu_{AC}} w(e),
  \end{equation}
  where $\mu_{ij}=\{(u,v)\in E: u\in \Gamma_j, v\in \Gamma_j\}$ is the set of edges that lies on the boundary between $\Gamma_i$ and $\Gamma_j$.}
We demonstrate a minimal configuration for the R\'enyi reflected entropy, as well as the negativity in Figure \ref{fig:TN_cuts2}.

\begin{figure}[th]
  \centering
  \includegraphics[width=.8\linewidth]{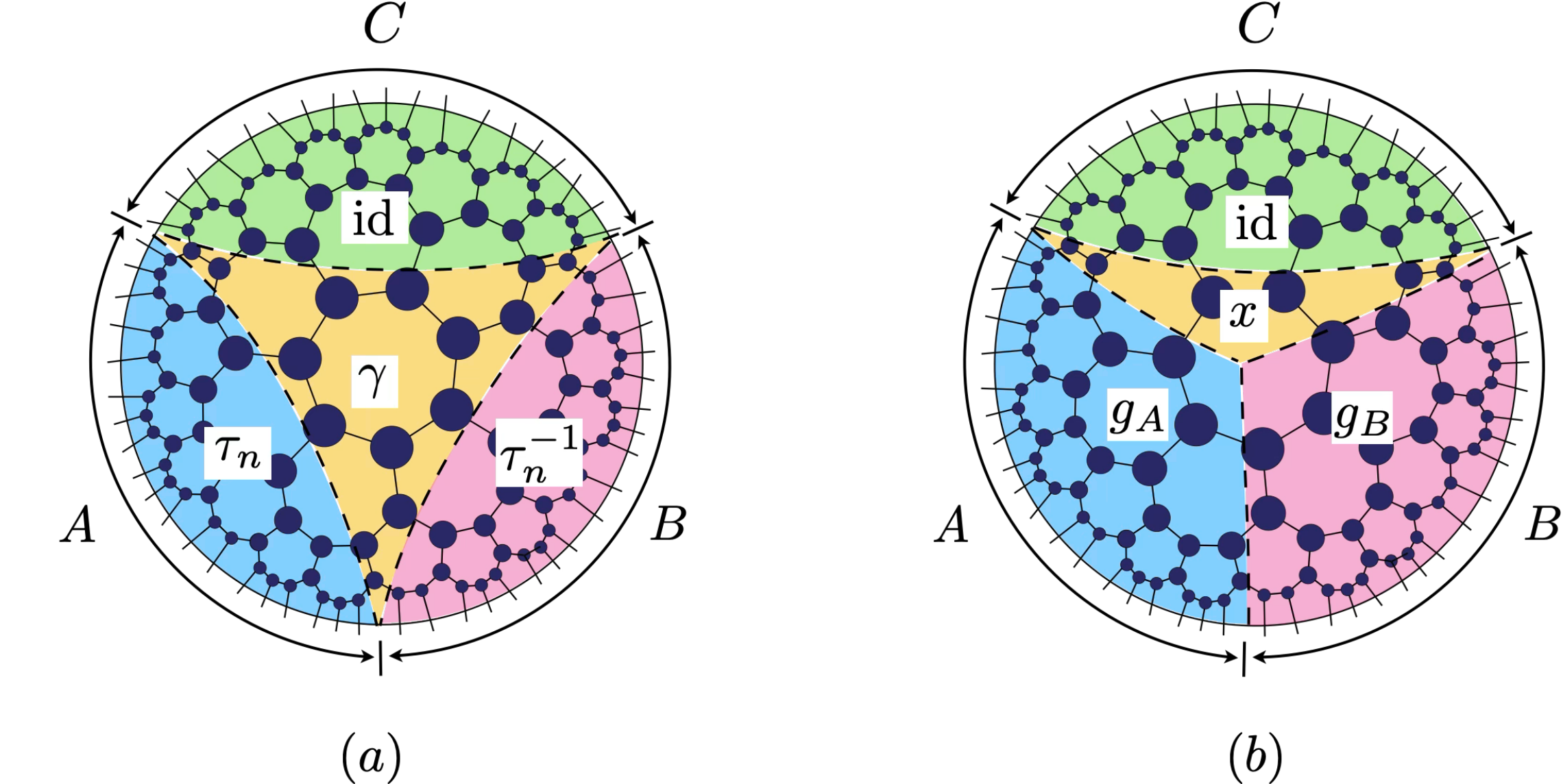}
  \caption{Examples for the minimal energy configuration of (a) R\'enyi negativity, and (b) R\'enyi reflected entropy in a hyperbolic tensor network. Different domains are represented by different colors.}
  \label{fig:TN_cuts2}
\end{figure}

From the above examples, we see that, in order to hope for a simple description for the minimizer in terms of the multiway cut, the geodesic structure between the terminal permutations $g_i$ must be ``highly frustrated''. In particular, the geodesics between any pairs of the terminal permutations had better be as far apart from each other as possible (see Proposition~\ref{prop:geodesic} for a result in this direction, showing that geodesics between terminal permutations intersect minimally).  For multi-entropy at $(q,n)=(3,2)$, Ref. \cite{Penington:2022dhr} showed that the minimal configuration is indeed given by the minimal tri-way cut. In Section \ref{sec:n=2} we extend their result to arbitrary integer $q$ (but keeping $n=2$), where instead of a minimal tri-way cut one has a minimal multiway cut across the network. For higher $n$, we show that the conjecture is, nonetheless, \emph{not true}, by giving an explicit counter example.

\section{A full description of the \texorpdfstring{$n=2$}{n=2} case}
\label{sec:n=2}
We now prove our main result of this paper, that the multi-entropy of RTN states are given by the area of the multiway cut in the network. Our proof applies to $n=2$, but arbitrary $q$. The higher $n$ cases are discussed in Section \ref{sec:n>2}.
Our proof rests on a simple estimate relating the Cayley distance and the number of moved points in $\Sym_N$. This estimate has previously appeared in the random matrix and quantum information theory literature, see e.g. \cite{biane1998representations,collins2011gaussianization, collins2013low}.
\begin{lemma}
  \label{lem:m(g)}
  For any permutation $g\in \Sym(X)$, let
  \begin{equation}
    m(g) := |\{x\in X: g(x) \ne x\}|
  \end{equation}
  be the number of the moved points under the action of $g$. Then
  \begin{equation}
    d(g,\id)\ge\frac{m(g)}{2}.
  \end{equation}
\end{lemma}
\begin{proof}
  Write the disjoint cycle decomposition of $g$ as fixed points together with nontrivial cycles of lengths
  \begin{equation}
    \ell_1,\cdots,\ell_r\ge 2.
  \end{equation}
  A cycle of length $\ell_j$ has transposition length $\ell_j-1$, hence
  \begin{equation}
    d(g,\id) = \sum_{j=1}^r(\ell_j-1).
  \end{equation}
  Also,
  \begin{equation}
    m(g) = \sum_{j=1}^r \ell_j
  \end{equation}
  because precisely the points lying in the nontrivial cycles are moved. Therefore
  \begin{equation}
    d(g,\id) -\frac{m(g)}{2} = \sum_{j=1}^r \left((\ell_j-1)-\frac{\ell_j}{2}\right) = \sum^r_{j=1} \left(\frac{\ell_j}{2}-1\right) \ge 0
  \end{equation}
  since each $\ell_j\ge2$.
\end{proof}
\begin{remark}
  \label{rem:m(g)}
  Note that the bound in the above lemma is not tight in general: for a full cycle $\tau$ on $X$, we get
  \begin{equation}
    d(\tau,\id) = |X|-1 \quad \text{while}\quad \frac{m(\tau)}{2} = \frac{|X|}{2},
  \end{equation}
  but it is tight if every nontrivial cycle has length $2$, i.e. when $g$ is an involution.
\end{remark}

\subsection{Single random tensor}\label{sec:n=2 single random tensor}
We first consider the situation for the case of a single random tensor before moving on to RTNs defined on generic graphs. That is, we consider the minimization of the following target
\begin{equation}
  \min_{g\in \Sym(X)} F(g) = \min_{g\in \Sym(X)} \sum_{i=1}^q w_i d(g,g_i),
\end{equation}
where $w_i \in \mathbb{R}_+$. See Figure \ref{fig:1TN} for an illustration. Let us first consider the case where $w_i$ are all equal to each other.

\begin{figure}[t]
    \centering
    \includegraphics[scale=.3]{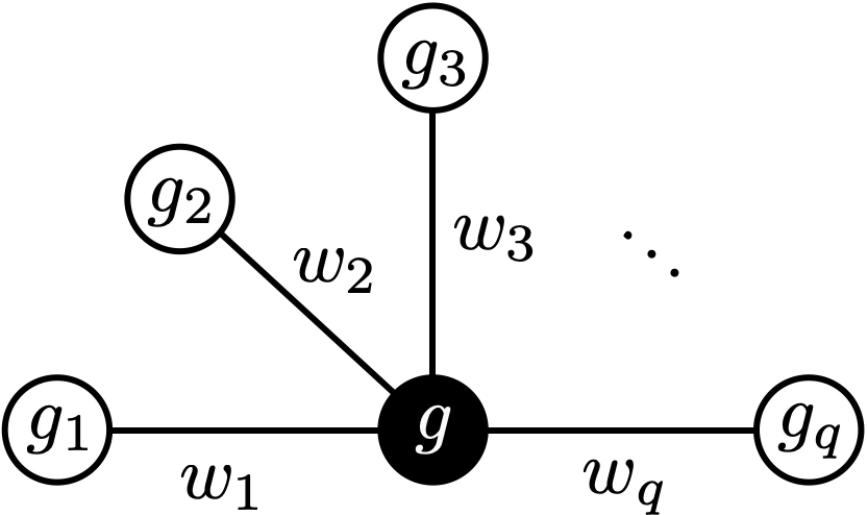}
    \caption{The $q$-partite single random tensor we consider in this subsection. The black node represents the random tensor and the white nodes represent the terminal vertices. The (logarithmic) bond dimensions of this random tensor are denoted by $w_i$, respectively.}
    \label{fig:1TN}
\end{figure}

\begin{theorem}
  \label{thm:1TN}
  For $n=2$ and for every integer $q$,
  \begin{equation}
    \min_{g\in \Sym(X)} \sum^q_{i=1} d(g,g_i) = (q-1)2^{q-2}.
  \end{equation}
  Moreover, the minimum is attained at $g=g_q=\id$.
\end{theorem}
\begin{proof}
  Fix an arbitrary permutation $g\in \Sym(X)$. For each $i=1,\cdots,q$, define
  \begin{equation}
    h_i = g^{-1}g_i.
  \end{equation}
  Then
  \begin{equation}
    d(g,g_i) = d(g^{-1}g_i,\id) = d(h_i,\id).
  \end{equation}
  By Lemma~\ref{lem:m(g)},
  \begin{equation}
    d(h_i,\id) \ge \frac{m(h_i)}{2}.
  \end{equation}
  Now
  \begin{equation}
    h_i(x)=x \Longleftrightarrow g^{-1}g_i(x) = x \Longleftrightarrow g_i(x) = g(x),
  \end{equation}
  so
  \begin{equation}
    \label{eq:1TN_cp}
    m(h_i) = |\{x\in X: g(x)\ne g_i(x)\}|.
  \end{equation}
  Summing over $i$, we obtain
  \begin{equation}
        \label{eq:1TN_cp2}
    \sum_i d(g,g_i) \ge \frac{1}{2} \sum_i|x\in X: g(x) \ne g_i(x)| = \frac{1}{2} \sum_{x\in X}\sum_i \mathbf{1}_{g(x)\ne g_i(x)}
  \end{equation}

  Fix $x\in X$. We claim that the $q$ elements $g_1(x),\cdots,g_q(x)$ are pairwise distinct. Indeed:
  \begin{itemize}
  \item $g_q(x)=x$;
  \item for each $q>i\ge1$, $g_i(x)$ differs from $x$ in exactly the $i$-th coordinate;
  \item if $i\ne j$, then $g_i(x)$ and $g_j(x)$ differ in both the $i$-th and $j$-th coordinates.
  \end{itemize}
  Hence the single value $g(x)$ can coincide with at most one of the values $g_1(x),\cdots,g_q(x)$. It follows that
  \begin{equation}
    \sum^q_{i=1} \mathbf{1}_{g(x)\ne g_i(x)} \ge q-1 \quad \text{for every } x\in X.
  \end{equation}
  Summing over all $x\in X$, we get
  \begin{equation}
    \sum_i d(g,g_i) \ge \frac{1}{2}\sum_{x\in X} (q-1) = \frac{|X|}{2} (q-1) = (q-1)2^{q-2}.
  \end{equation}
  Since $g\in \Sym(X)$ was arbitrary, this proves
  \begin{equation}
    \min_{g\in \Sym(X)} \sum_{i=1}^q d(g,g_i) \ge (q-1)2^{q-2}.
  \end{equation}

  It remains to show that the equality is attained. Take $g=g_q=\id$.
  Then $d(g_q,g_q)=0$ and $d(g_q,g_i)=2^{q-2}$ for each $i=1,\cdots,q-1$. Therefore
  \begin{equation}
    \sum_i d(g_q,g_i) = (q-1)2^{q-2}.
  \end{equation}
  Combining this with the lower bound yields
  \begin{equation}
    \min_{g\in \Sym(X)} \sum_{i=1}^q d(g,g_i) = (q-1)2^{q-2}.
  \end{equation}
   and thus completes the proof.
\end{proof}

\begin{corollary}
  The set of minimizers in Theorem \ref{thm:1TN} is exactly
  \begin{equation}
    \{g_1,g_2,\cdots,g_q\}.
  \end{equation}
\end{corollary}
\begin{proof}
  Let $g$ be a minimizer. Then equality must hold at every step of the chain of inequalities in the proof of the theorem:
  \begin{equation}
    \sum_i d(g,g_i)\ge \frac{1}{2}\sum_i m(g^{-1}g_i) = \frac{1}{2}\sum_{x\in X}\sum_i \mathbf{1}_{g(x)\ne g_i(x)} \ge (q-1)2^{q-2}.
  \end{equation}
  Therefore, for every $x\in X$,
  \begin{equation}
    \sum_i \mathbf{1}_{g(x)\ne g_i(x)} = q.
  \end{equation}
  Since the values $g_1(x),\cdots,g_q(x)$ are pairwise distinct, it follows that $g(x)$ coincides with exactly one of them.

  Also, for every $i$, we have
  \begin{equation}
    d(g,g_i) = \frac{1}{2}m(g^{-1}g_i).
  \end{equation}
  From Remark \ref{rem:m(g)}, equality occurs if and only if $g^{-1}g_i$ is an involution. Taking $i=q$, we see that $g^{-1}$ is an involution, so $g^2=\id$ and therefore $g=g^{-1}$. For $i=1,\cdots,q-1$, the permutation $g^{-1}g_i=gg_i$ is also an involution, therefore
  \begin{equation}
    g_i g g_i^{-1} = g,
  \end{equation}
  since $g_i$ are also involutions. This implies that $g$ commutes with every translation: for any $v\in X$, then $g(\tau_v(x))=\tau_v(g(x))$ where
  \begin{equation}
    \tau_v: X\to X, \quad \tau_v(x) = x+v.
  \end{equation}

  Finally, from the pointwise condition above, we have $g(0)=g_i(0)$ for some $i\in\{1,\cdots,q\}$. We have
 For any $x\in X$,
  \begin{equation}
    g(x) = g(\tau_x(0)) = \tau_x(g(0)) = x + g_i(0) = g_i(x).
  \end{equation}
  So $g=g_i$. This proves that the set of minimizers is exactly $\{g_1,\cdots,g_q\}$.
\end{proof}

\begin{remark}
  In the proof above it is crucial that the $g=g_i$ saturates the estimate $d(g,g_j)\ge\frac{1}{2}m(h_j)=\frac{1}{2}m(g^{-1}g_j)$.
  This is only possible if $h_j$ is an involution, i.e. $h_j^2=\id$, which happens to be true when $n=2$. For $n\ge3$, $(g_i^{-1}g_j)^2\ne \id$ and the estimate is never tight for $g=g_i$. In fact we will see that, for $n\ge 3$ there is a new element $\pi$ where it attains a lower objective, $F(g)\ge F(\pi)$ (with equality only at $(q, n)=(3,3)$).
\end{remark}

We now deal with the case of arbitrary edge weights.
\begin{theorem}
  \label{thm:1TN_w}
  Let $w_1,\cdots,w_q>0$, and set
  \begin{equation}
    w_{\max}:=\max_{1\le i\le q}w_i,\qquad
    K_{\max}:=\{k\in\{1,\cdots,q\}:w_k=w_{\max}\}.
  \end{equation}
  For
  \begin{equation}
    F(g) = \sum_{i=1}^q w_i d(g,g_i)
  \end{equation}
  one has
  \begin{equation}
    \min_{g\in \Sym(X)}F(g)
    =
    2^{q-2}\left(\Big(\sum_{i=1}^q w_i \Big)-w_{\max}\right),
  \end{equation}
  and the minimizers are exactly
  \begin{equation}
    \{g_k:k\in K_{\max}\}.
  \end{equation}
\end{theorem}
\begin{proof}
  We follow the same proof as that of Theorem \ref{thm:1TN} up to \eqref{eq:1TN_cp}. However instead of \eqref{eq:1TN_cp2} we have
  \begin{equation}
    \sum_i w_i d(g,g_i) \ge \frac{1}{2}\sum_i w_i|x\in X:g(x)\ne g_i(x)|=\frac{1}{2}\sum_{x\in X}\sum_i w_i\mathbf{1}_{g(x)\ne g_i(x)}.
  \end{equation}
  Now since $g(x)$ can only coincide with at most one of $g_1(x),\cdots,g_q(x)$, to minimize the RHS we must pick $g(x)=g_k(x)$ for some $k\in K_{\max}$, and thus
  \begin{equation}
    \sum_i w_i \mathbf{1}_{g(x)\ne g_i(x)} \ge \Big(\sum_{i=1}^q w_i \Big)-w_{\max} , \quad \forall x\in X.
  \end{equation}
  Summing over all $x\in X$ we obtain
  \begin{equation}
    \sum_i w_i d(g,g_i) \ge 2^{q-2}\left(\Big(\sum_{i=1}^q w_i \Big) -w_{\max}\right).
  \end{equation}
  It is straightforward to check that plugging $g=g_k$ into $F(g)$ saturates this bound for every $k\in K_{\max}$.

  Conversely, let $g$ be a minimizer. Then equality holds in the pointwise bound above, so for every $x\in X$ there is some $k(x)\in K_{\max}$ such that
  \begin{equation}
    g(x)=g_{k(x)}(x).
  \end{equation}
  Also equality holds in the bounds $d(g,g_i)\ge \frac{1}{2}m(g^{-1}g_i)$ for all $i$, hence every $g^{-1}g_i$ is an involution. As in the proof of the corollary to Theorem \ref{thm:1TN}, this implies that $g$ commutes with every translation of $X$. Since $g(0)=g_k(0)$ for some $k\in K_{\max}$, we have, for every $x\in X$,
  \begin{equation}
    g(x)=g(\tau_x(0))=\tau_x(g(0))=\tau_x(g_k(0))=g_k(x).
  \end{equation}
  Thus $g=g_k$ with $k\in K_{\max}$, which proves the claimed description of the minimizers.
\end{proof}

\subsection{Random tensor networks}
We now deal with the general optimization problem of multi-entropy on a RTN state:
Given a graph $G=(V,E)$ and edge weight $w:E\to \mathbb{R}_+$, we would like to minimize
\begin{equation}
  \label{eq:RTN_F(g)}
  F(g) = \sum_{e=(u,v)\in E} w(e) d(g(u),g(v))
\end{equation}
over all admissible $g:V\to \Sym(X)$ satisfying the boundary condition $g(v)=g_i$ for $v\in R_i$ on the boundary (terminal) vertices.

We first state our main result of this subsection.
\begin{theorem}
  \label{thm:RTN}
  For $n=2$ and for every integer $q$, the minimum of \eqref{eq:RTN_F(g)} is given by the minimum multiway cut area over $G$:
  \begin{equation}
    \min_g F(g) = 2^{q-2}\mathcal{A}(R_1:\cdots :R_q).
  \end{equation}
\end{theorem}
Our proof strategy will be similar to the single tensor case, by first applying Lemma \ref{lem:m(g)} on the Cayley distances to obtain a bound for $F(g)$. We then show that, given a solution to the multiway cut problem, one can reconstruct a $g$ that saturates the bound.

For each $x\in X$, we define the \emph{slice constant}
\begin{equation}
  \mathcal{A}_x := \min_s\left\{ \sum_{e=(u,v)\in E} w(e) \mathbf{1}_{s(u)\ne s(v)}: s:V\to X, \quad s(v)=g_i(x)\quad  \forall v\in R_i \right\} .
\end{equation}
Equivalently, $\mathcal{A}_x$ is the (weighted) minimum number of edges one must cut in order to separate the terminal nodes whose $x$-images under $g_i$ are different, while terminals with the same $x$-image are allowed to remain in the same component. Note that this quantity depends on the element $x\in X$.

\begin{lemma}[Universal slice lower bound]
  \label{lem:slice}
  For every admissible $g$,
  \begin{equation}
    \label{eq:slice_ineq}
    F(g) \ge \frac{1}{2} \sum_{x\in X} \mathcal{A}_x.
  \end{equation}
\end{lemma}
\begin{proof}
  Fix an admissible $g$. For each edge $(u,v)\in E$ denote $g_u=g(u)$ and $g_v=g(v)$. Lemma \ref{lem:m(g)} gives
  \begin{equation}
    d(g_u,g_v) \ge \frac{1}{2}|\{x\in X: g_u(x) \ne g_v(x)\}|,
  \end{equation}
  Summing over all edges and interchanging the order of summation, we obtain
  \begin{align}
    \label{eq:slice_ineq2}
    F(g) &\ge \frac{1}{2} \sum_{e=(u,v)\in E} w(e) |\{x\in X: g_u(x)\ne g_v(x)\}| \\
    &= \frac{1}{2}\sum_{x\in X}\sum_{e=(u,v)\in E} w(e) \mathbf{1}_{g_u(x)\ne g_v(x)}.
  \end{align}

  For a fixed $x\in X$, define
  \begin{equation}
    s_x(v) := g_v(x), \quad (v\in V).
  \end{equation}
  Then $s_x:V\to X$ is a labeling of $V$ with boundary values
  \begin{equation}
    s_x(t_i) = g_{t_i}(x) = g_i(x), \quad t_i\in R_i.
  \end{equation}
  Therefore, by the definition of $\mathcal{A}_x$,
  \begin{equation}
    \sum_{e=(u,v)\in E} w(e)\mathbf{1}_{g_u(x)\ne g_v(x)} = \sum_{e=(u,v)\in E} w(e)\mathbf{1}_{s_x(u)\ne s_x(v)} \ge \mathcal{A}_x.
  \end{equation}
  Substituting this into \eqref{eq:slice_ineq2} and summing over $x$ gives \eqref{eq:slice_ineq}.
\end{proof}

We need the following technical lemma.
\begin{lemma}[disjoint labelings are multiway cuts]
  \label{lem:labeling}
  Let $Y$ be any set, and let $y_1,\cdots,y_q\in Y$ be pairwise distinct. For a labeling $s:V\to Y$ with
  \begin{equation}
    s(v) = y_i,\quad v\in R_i.
  \end{equation}
  Define
  \begin{equation}
    \mathcal{W}(s) := \sum_{e=(u,v)\in E} w(e) \mathbf{1}_{s(u)\ne s(v)}.
  \end{equation}
  Then
  \begin{equation}
    \label{eq:labeling_ineq}
    \mathcal{W}(s) \ge \mathcal{A}(R_1:\cdots:R_q).
  \end{equation}
  Moreover, equality is attained by the boundary-valued labeling: if $V=\Gamma_1\sqcup\cdots\sqcup \Gamma_q$ is a minimum multiway partition, then the labeling
  \begin{equation}
    \label{eq:labeling_part}
    s(v) = y_i,\quad v\in \Gamma_i
  \end{equation}
  satisfies
  \begin{equation}
    \mathcal{W}(s) = \mathcal{A}(R_1:\cdots:R_q).
  \end{equation}
\end{lemma}
\begin{proof}
  Let
  \begin{equation}
    E_{\ne} (s) := \{(u,v)\in E: s(u)\ne s(v)\}.
  \end{equation}
  Then $\mathcal{W}(s) = \sum_{e\in E_{\ne}(s)}w(e)$. Remove these edges and consider the spanning subgraph
  \begin{equation}
    G_s = (V,E\setminus E_{\ne}(s)).
  \end{equation}
  Every connected component of $G_s$ is monochromatic by construction. Since the terminal colors $y_i$ are pairwise distinct, no connected component of $G_s$ can contain boundary vertices of different indices.

  For each $i$, let $\Gamma_i$ be the union of all components of $G_s$ that intersect $R_i$. Any remaining component of $G_s$ contains no terminal vertices; assign each such component arbitrarily to one of the sets $\Gamma_1,\cdots,\Gamma_q$. The produces a multiway partition
  \begin{equation}
    V= \Gamma_1 \sqcup\cdots\sqcup \Gamma_q.
  \end{equation}
  Every edge joining different parts of this partition must lie in $E_{\ne}(s)$, because edges outside $E_{\ne}(s)$ stay inside components of $G_s$. Hence
  \begin{equation}
    \mu(\{\Gamma_k\}) \subseteq E_{\ne} (s).
  \end{equation}
  So
  \begin{equation}
    \label{eq:labeling_ineq2}
    \mathcal{W}(s) = \sum_{e\in E_{\ne}(s)} w(e) \ge \sum_{e\in \mu(\{\Gamma_k\})} w(e) \ge \mathcal{A}(R_1:\cdots :R_q).
  \end{equation}
  This proves \eqref{eq:labeling_ineq}.

  Conversely, if $V=\Gamma_1\sqcup\cdots\sqcup \Gamma_q$ is a minimum multiway partition and we set $s(v)=y_i$ for $v\in \Gamma_i$, then an edge contributes $w(e)$ to $\mathcal{W}(s)$ if and only if it lies in $\mu(\{\Gamma_k\})$. Therefore
  \begin{equation}
    \mathcal{W}(s) = \sum_{e\in \mu(\{\Gamma_k\})} w(e) = \mathcal{A}(R_1:\cdots :R_q).
  \end{equation}
\end{proof}

\begin{remark}
  \label{rem:labeling}
  Note that \eqref{eq:labeling_ineq} is saturated \emph{if and only if} $s$ is of the form of \eqref{eq:labeling_part}. This can be seen from \eqref{eq:labeling_ineq2}: Saturating of the chain of inequalities requires $E_{\ne}(s)=\mu(\{\Gamma_k\})$ where $\{\Gamma_k\}$ is a minimal multiway partition. We call such $s$ to be a labeling \emph{associated} to the minimal multiway partition $\{\Gamma_k\}$.
\end{remark}

We are now ready to prove the main result.

\begin{proof}[Proof of Theorem \ref{thm:RTN}]
  Fix $x\in X$. The $q$ elements $g_1(x),\cdots,g_q(x)$ are pairwise distinct for $n=2$ (see the proof of Theorem \ref{thm:1TN}).
  Therefore Lemma~\ref{lem:labeling}, applied to the colors $y_i=g_i(x)$ shows that the minimum disagreement cost with these boundary values is exactly the minimum multiway cut:
  \begin{equation}
    \mathcal{A}_x \ge \mathcal{A}(R_1:\cdots:R_q).
  \end{equation}
  From Lemma \ref{lem:slice} we have
  \begin{equation}
    F(g) \ge \frac{1}{2}\sum_{x\in X} \mathcal{A}_x \ge \frac{|X|}{2} \mathcal{A}(R_1:\cdots:R_q) = 2^{q-2} \mathcal{A}(R_1:\cdots:R_q).
  \end{equation}

  Next, let $V=\Gamma_1\sqcup\cdots \sqcup \Gamma_q$ be a minimum multiway partition, and define
  \begin{equation}
    g_*(v)=g_i,\quad v\in \Gamma_i.
  \end{equation}
  It is clear that $g_*$ is admissible. The only contribution to $F(g_*)$ comes from the cut set $\mu(\{\Gamma_k\})$, with each edge contributing $d(g_i,g_j)=2^{q-2}$. Therefore
  \begin{equation}
    F(g_*) = 2^{q-2} \sum_{e\in \mu(\{\Gamma_k\})} w(e) = 2^{q-2} \mathcal{A} (R_1:\cdots:R_q)
  \end{equation}
  saturates the bound. This proves the theorem.
\end{proof}

In fact, we can state the previous result of single tensor as a simple corollary.
\begin{example}[single random tensor]
  Let $G=(V,E)$ where $V$ consists of a single internal vertex $v$ and terminals $\{t_1,\cdots,t_q\}$, and edge weight $w(v,t_i) = w_i$.
  Let $k\in\{1,\cdots,q\}$ such that $w_k = \max_i w_i$.   The minimal multiway partition of this network is
  \begin{equation}
    \Gamma_{i\ne k} = \{t_i\}, \quad \Gamma_k = \{v,t_k\}.
  \end{equation}
  The area of the minimal multiway cut is
  \begin{equation}
    \mathcal{A}(t_1:\cdots:t_q) = \sum_{i\ne k} w_i.
  \end{equation}
  Applying Theorem \ref{thm:RTN}, we find
  \begin{equation}
    \min_g F(g) = 2^{q-2}\mathcal{A}(t_1:\cdots:t_q) = 2^{q-2}\sum_{i\ne k} w_i,
  \end{equation}
  which reproduces the result of Theorem \ref{thm:1TN_w}.
\end{example}

The next result pertains to the uniqueness of the minimizer.
\begin{corollary}\label{cor:unique-mincut}
  If $G$ admits a unique minimal multiway partition $V=\Gamma_1\sqcup\cdots\sqcup \Gamma_q$, then the minimizer of $F(g)$ is unique and is given by
  \begin{equation}
    g(v) = g_i, \quad v\in \Gamma_i
  \end{equation}
\end{corollary}
\begin{proof}
  Since $g$ is a minimizer, it must saturate the chain of inequalities in the proof of Theorem \ref{thm:RTN}:
  \begin{equation}
    \label{eq:RTN_unique_ineq}
    F(g) \ge\frac{1}{2}\sum_{x\in X}\sum_{e=(u,v)\in E} w(e) \mathbf{1}_{g_u(x)\ne g_v(x)} \ge \frac{1}{2}\sum_{x\in X}\mathcal{A}_x \ge \frac{1}{2}\sum_{x\in X} \mathcal{A}(R_1:\cdots:R_q)
  \end{equation}
  In particular this chain of inequalities has to be tight for each $x\in X$.

  Fix some $x\in X$. Let $s_x:V\to X$ be a labeling defined by
  \begin{equation}
    s_x(v) = g_v(x).
  \end{equation}
  By Remark \ref{rem:labeling}, saturation of \eqref{eq:RTN_unique_ineq} requires $s_x$ to be the labeling associated to the unique minimal multiway partition $\{\Gamma_k\}$, i.e.
  \begin{equation}
    s_x(v) = g_i(x), \quad v\in \Gamma_i,
  \end{equation}
  which implies that for any two vertices $u,v$ in the same part of the multiway partition $\{\Gamma_k\}$
  \begin{equation}
    g_u(x)=g_v(x), \quad u,v\in \Gamma_i.
  \end{equation}
  Since this is true for all $x$, we must have
  \begin{equation}
    g(u)=g(v), \quad u,v\in \Gamma_i.
  \end{equation}
  Now take $u\in R_i$. Since $g$ is admissible, $g(u)=g_i$. This proves that
  \begin{equation}
    g(v) = g_i, \quad v\in \Gamma_i.
  \end{equation}
\end{proof}

When the minimal multiway partition is not unique, then each of such partition gives rise to a solution of the multi-entropy minimization problem. The question is whether these solutions constitute all the minimizers of the problem. This problem, as well as several examples, will be discussed in the next section.

\section{Minimizer count and examples}
\label{sec:minimizer}

In this section we count the number and study the nature of the minimizers of the function $F$. We discuss some simple examples that allow for minimizers which take values outside the set of terminal permutations. We emphasize that trees are easy to analyze and we study the simplest non-tree network with 3 terminals, the triangle network.

\subsection{Counting minimizers in the general case}

We now describe all minimizers in the connected case.  Throughout this
subsection we keep $n=2$, assume that $G$ is connected and that all edge
weights are strictly positive.  We write
\begin{equation}
  \mathcal{A}:=\mathcal{A}(R_1:\cdots:R_q).
\end{equation}
Let $e_1,\cdots,e_{q-1}$ be the standard generators of
$X=\mathbb{Z}_2^{q-1}$ and put $e_q=0$, so that
\begin{equation}
  g_i(x)=x+e_i,\quad i=1,\cdots,q,
\end{equation}
with $g_q=\id$.  We also write $\operatorname{Min}(F)$ for the set of
admissible minimizers of \eqref{eq:RTN_F(g)}; the goal of this
subsection is to provide a counting formula for the number of minimizers. We start with the precise necessary and sufficient conditions for equality in Thm.~\ref{thm:RTN}.

\begin{proposition}
  \label{prop:counting-equality-conditions}
  An admissible labeling $g:V\to \Sym(X)$ minimizes \eqref{eq:RTN_F(g)} if
  and only if the following two conditions hold.
  \begin{enumerate}[label=(\arabic*)]
  \item \label{item:counting-edge-involution}
    For every edge $e=(u,v)\in E$, the relative permutation
    $g(u)^{-1}g(v)$ satisfies
    \begin{equation}
      (g(u)^{-1}g(v))^2=\id.
    \end{equation}

  \item \label{item:counting-slice-mincut}
    For every $x\in X$, the slice
    \begin{equation}
      s_x(v)=g(v)(x)
    \end{equation}
    is terminal-valued on a minimal multiway partition: there exists
    $V=\Gamma_1^x\sqcup\cdots\sqcup\Gamma_q^x$ with area $\mathcal{A}$
    such that
    \begin{equation}
      s_x(v)=g_i(x),\quad v\in\Gamma_i^x.
    \end{equation}
  \end{enumerate}
\end{proposition}
\begin{proof}
  Write $g_u=g(u)$ and $g_v=g(v)$.  For every edge $e=(u,v)$, Lemma
  \ref{lem:m(g)} gives
  \begin{equation}
    \label{eq:counting-edge-lower}
    d(g_u,g_v)\ge
    \frac{1}{2}|\{x\in X:g_u(x)\ne g_v(x)\}|.
  \end{equation}
  Summing over all edges and exchanging the order of summation,
  \begin{align}
    F(g)
    &\ge \frac{1}{2}\sum_{e=(u,v)\in E} w(e)
      |\{x\in X:g_u(x)\ne g_v(x)\}| \notag\\
    &=\frac{1}{2}\sum_{x\in X}\sum_{e=(u,v)\in E}
      w(e)\mathbf{1}_{g_u(x)\ne g_v(x)}. \label{eq:counting-slice-lower}
  \end{align}
  For fixed $x$, the boundary values of $s_x(v)=g(v)(x)$ are
  $g_i(x)$ on $R_i$.  These $q$ values are pairwise distinct for
  $n=2$.  Applying Lemma \ref{lem:labeling} to the slice $s_x$, we get
  \begin{equation}
    \sum_{e=(u,v)\in E} w(e)\mathbf{1}_{g_u(x)\ne g_v(x)}
    \ge \mathcal{A}.
  \end{equation}
  Therefore
  \begin{equation}
    F(g)\ge \frac{|X|}{2}\mathcal{A}=2^{q-2}\mathcal{A}.
  \end{equation}
  By Theorem \ref{thm:RTN}, this is the minimum value.

  If $g$ is a minimizer, every nonnegative error in the two inequalities
  above must vanish.  Since all edge weights are positive, equality in
  \eqref{eq:counting-edge-lower} holds for every edge.  By the proof of
  Lemma \ref{lem:m(g)}, this is equivalent to every nontrivial cycle of
  $g(u)^{-1}g(v)$ having length $2$, or equivalently
  $(g(u)^{-1}g(v))^2=\id$.  This gives condition
  \ref{item:counting-edge-involution}.

  Equality in the slice lower bound holds for every $x\in X$.  Remark
  \ref{rem:labeling} then shows that each slice is
  terminal-valued on a minimal multiway partition.  This gives condition
  \ref{item:counting-slice-mincut}.

  Conversely, if conditions \ref{item:counting-edge-involution} and
  \ref{item:counting-slice-mincut} hold, then every edge inequality
  \eqref{eq:counting-edge-lower} is an equality and every slice has
  disagreement cost $\mathcal{A}$.  Hence
  \begin{equation}
    F(g)=\frac{1}{2}\sum_{x\in X}\mathcal{A}
    =2^{q-2}\mathcal{A},
  \end{equation}
  so $g$ is a minimizer.
\end{proof}

Since we are considering in this section the case of arbitrary tensor networks where the minimal multiway cut might not be unique, we introduce next the set of all such minimal cuts.

\begin{definition}
  \label{def:cut-geodesic-data}
  Let $\mathfrak{M}$ be the set of \emph{minimal multiway cut labelings}
  \begin{equation}
    \lambda:V\to\{1,\cdots,q\},\quad
    \lambda(v)=i\quad (v\in R_i),
  \end{equation}
  whose cut area is $\mathcal{A}$.  Equivalently, $\lambda$ records a
  minimal multiway partition by
  \begin{equation}
    \Gamma_i=\lambda^{-1}(i).
  \end{equation}

  For each pair $i\ne j$, define the terminal-geodesic matching
  \begin{equation}\label{eq:terminal-geodesic-matching}
    \mathcal{P}_{ij}:=\{\{x,x+e_i+e_j\}:x\in X\}.
  \end{equation}
  This is exactly the set of transpositions appearing in the disjoint
  cycle decomposition of $g_i^{-1}g_j$.  Consequently the Cayley
  geodesic interval $[g_i,g_j]$ consists of the permutations obtained
  from $g_i$ by choosing an arbitrary subcollection of the pairs in
  $\mathcal{P}_{ij}$ and applying those transpositions.
\end{definition}

The next theorem is the main result of this section, characterizing
exactly the set of minimizers in terms of two combinatorial conditions on
functions $X \to \mathfrak M$.  Thus, for each $x\in X$, we choose a
minimal multiway cut labeling
\begin{equation}
  \lambda_x:V\to\{1,\cdots,q\}.
\end{equation}
Equivalently, $\lambda_x$ determines a partition of $V$ into $q$ blocks
\begin{equation}
  \Gamma_i(x):=\lambda_x^{-1}(i),\qquad i=1,\cdots,q,
\end{equation}
where the block $\Gamma_i(x)$ contains the terminal region $R_i$.  We
write
\begin{equation}
  \Lambda=(\lambda_x)_{x\in X}\in\mathfrak{M}^X
\end{equation}
for the resulting family of minimal multiway cut labelings.

\begin{theorem}[Characterization of minimizers]
  \label{thm:compatible-cut-geodesic-families}
  For a family $\Lambda=(\lambda_x)_{x\in X}\in\mathfrak{M}^X$, define
  pointwise maps
  \begin{equation}
    \label{eq:compatible-family-g}
    g_v^\Lambda(x):=g_{\lambda_x(v)}(x)=x+e_{\lambda_x(v)}.
  \end{equation}
  Then the maps $g_v^\Lambda$ define an admissible minimizer
  $g^\Lambda:V\to \Sym(X)$ if and only if the following two combinatorial
  conditions hold:

  \begin{enumerate}
  \item[\textup{(V)}] \namedlabel{item:compatible-vertex}{(V)}
    For every vertex $v$, the map
    \begin{equation}
      x\longmapsto x+e_{\lambda_x(v)}
    \end{equation}
    is a bijection of $X$.  Equivalently, for every $y\in X$,
    \begin{equation}
      \label{eq:vertex-bijection-condition}
      |\{i\in\{1,\cdots,q\}:\lambda_{y+e_i}(v)=i\}|=1.
    \end{equation}

  \item[\textup{(E)}] \namedlabel{item:compatible-edge}{(E)}
    For every edge $e=(u,v)\in E$, every $x\in X$, and
    \begin{equation}
      i=\lambda_x(u),\quad j=\lambda_x(v),
    \end{equation}
    one has
    \begin{equation}
      \label{eq:edge-geodesic-compatibility}
      \lambda_{x+e_i+e_j}(u)=i,\quad
      \lambda_{x+e_i+e_j}(v)=j.
    \end{equation}
  \end{enumerate}
  Equivalently, condition \ref{item:compatible-edge} says that, along
  the graph edge $(u,v)$, the relative permutation
  $(g_u^\Lambda)^{-1}g_v^\Lambda$ is a product of selected
  transpositions from the terminal geodesics $[g_i,g_j]$.

  Consequently the number of minimizers is given by
  \begin{equation}
    \label{eq:minimizer-count-compatible-families}
    |\operatorname{Min}(F)|
    = |\{\Lambda\in\mathfrak{M}^X \, : \, \Lambda \text{ satisfies \ref{item:compatible-vertex} and \ref{item:compatible-edge}}\}|.
  \end{equation}
\end{theorem}

\begin{proof}
  First let $g\in\operatorname{Min}(F)$.  By Proposition
  \ref{prop:counting-equality-conditions}, for every $x\in X$ the slice
  $s_x(v)=g(v)(x)$ is terminal-valued on a minimal multiway partition.
  Since the values $g_1(x),\cdots,g_q(x)$ are pairwise distinct, there
  is a unique $\lambda_x\in\mathfrak{M}$ such that
  \begin{equation}
    g(v)(x)=g_{\lambda_x(v)}(x),\quad v\in V.
  \end{equation}
  Thus $g=g^\Lambda$ for the family
  $\Lambda=(\lambda_x)_{x\in X}$.  Condition
  \ref{item:compatible-vertex} holds because
  each $g(v)$ is a permutation of $X$.

  We now derive condition \ref{item:compatible-edge}.  Fix an edge
  $e=(u,v)$ and put $f=g(u)$, $h=g(v)$. Proposition \ref{prop:counting-equality-conditions} gives $(f^{-1}h)^2=\id$.  Fix $x\in X$ and write $i=\lambda_x(u)$, $j=\lambda_x(v)$. If $i=j$, then \eqref{eq:edge-geodesic-compatibility} is immediate.
  Suppose $i\ne j$, and let $y=(f^{-1}h)(x)$. Then
  \begin{equation}
    f(y)=h(x)=x+e_j.
  \end{equation}
  Write $a=\lambda_y(u)$, $b=\lambda_y(v)$. Since $f(y)=y+e_a$, we get $y+e_a=x+e_j$. This is equivalent to $y=x+e_a+e_j$ because $n=2$ and $e_a=-e_a$. Since $(f^{-1}h)^2=\id$, we also have $h(y)=f(x)$, hence
  \begin{equation}
    y+e_b=x+e_i \quad \text{ and thus } \quad e_i+e_j=e_a+e_b.
  \end{equation}
  We now use the simplex property of $e_1, \ldots, e_{q-1}$. The vectors $e_1,\cdots,e_{q-1},e_q=0$ have \emph{unique two-term sums}: if $i\ne j$ and $e_i+e_j=e_a+e_b$, then
  $\{a,b\}=\{i,j\}$.  Also $y\ne x$, because $f(x)\ne h(x)$, so
  $a\ne j$. Hence
  \begin{equation}
    a=i,\quad b=j,\quad y=x+e_i+e_j,
  \end{equation}
  which is exactly \eqref{eq:edge-geodesic-compatibility}.

  \medskip

  Conversely, suppose $\Lambda\in\mathfrak{M}^X$ satisfies
  conditions \ref{item:compatible-vertex} and
  \ref{item:compatible-edge}.  By condition
  \ref{item:compatible-vertex}, each $g_v^\Lambda$ is a permutation of
  $X$.  Since every $\lambda_x$ has
  $\lambda_x(v)=k$ on $R_k$, the labeling $g^\Lambda$ is admissible.
  Each slice is terminal-valued on the minimal cut labeling $\lambda_x$,
  so condition \ref{item:counting-slice-mincut} of Proposition
  \ref{prop:counting-equality-conditions} holds.

  It remains to check condition \ref{item:counting-edge-involution}.
  Fix an edge $e=(u,v)$ and put $f=g_u^\Lambda,\quad h=g_v^\Lambda$. For $x\in X$, let $i=\lambda_x(u)$ and $j=\lambda_x(v)$.  If $i=j$,
  then $f(x)=h(x)$, so $f^{-1}h$ fixes $x$.  If $i\ne j$, condition
  \ref{item:compatible-edge} says that at $y=x+e_i+e_j$, the same ordered pair $(i,j)$ occurs. Thus
  \begin{equation}
    f(y)=y+e_i=x+e_j=h(x),\quad
    h(y)=y+e_j=x+e_i=f(x).
  \end{equation}
  Since $f$ is bijective, this means
  \begin{equation}
    (f^{-1}h)(x)=y,\quad (f^{-1}h)(y)=x.
  \end{equation}
  Thus $f^{-1}h$ is a product of fixed points and transpositions, so
  $(f^{-1}h)^2=\id$.  Condition
  \ref{item:counting-edge-involution} of Proposition
  \ref{prop:counting-equality-conditions} holds. Therefore
  $g^\Lambda$ is a minimizer.

  The formula \eqref{eq:minimizer-count-compatible-families} is this
  bijection written as an indicator sum over all $X$-indexed families of
  minimal cut labelings, or, equivalently, over all functions $X \to \mathfrak M$.
\end{proof}

We can now recover the unique multiway min-cut result as a corollary of
the previous theorem.
\begin{corollary}
  \label{cor:unique-mincut-from-characterization}
  Suppose that $G$ admits a unique minimal multiway cut labeling
  $\lambda:V\to\{1,\cdots,q\}$, and write
  $\Gamma_i=\lambda^{-1}(i)$.  Then the minimizer of $F(g)$ is unique
  and is given by
  \begin{equation}
    g(v)=g_i,\quad v\in\Gamma_i.
  \end{equation}
\end{corollary}
\begin{proof}
  Since $\mathfrak{M}=\{\lambda\}$, there is only one family
  $\Lambda=(\lambda_x)_{x\in X}\in\mathfrak{M}^X$, namely
  $\lambda_x=\lambda$ for all $x$.  Condition
  \ref{item:compatible-vertex} holds because
  $x\mapsto x+e_{\lambda(v)}$ is a translation of $X$, and condition
  \ref{item:compatible-edge} is immediate because the labeling is
  independent of $x$.  Theorem
  \ref{thm:compatible-cut-geodesic-families} therefore gives exactly one
  minimizer, and for it
  \begin{equation}
    g_v^\Lambda(x)=g_{\lambda(v)}(x).
  \end{equation}
  Thus $g(v)=g_{\lambda(v)}$, equivalently $g(v)=g_i$ on
  $\Gamma_i$.
\end{proof}

Let us consider now an example of a connected tree with non-terminal-valued minimizers. Let $q=3$, so that $X=\mathbb{Z}_2^2$.  We write the four elements of
  $X$ as $00$, $10$, $01$, and $11$, with $e_1=10$, $e_2=01$, and $e_3=0$.  Thus
  \begin{equation}
    g_1=(00\ 10)(01\ 11),\quad
    g_2=(00\ 01)(10\ 11),\quad
    g_3=\id .
  \end{equation}

  Consider the unit-weight tree with vertices $t_1$, $t_2$, $t_3$, $u$, and $v$,
  and edges
  \begin{equation}
    (t_3,u),\quad (t_2,u),\quad (u,v),\quad (v,t_1),
  \end{equation}
  with terminal regions $R_i=\{t_i\}$. In words, $u$ is adjacent to
  $t_2,t_3$ and to $v$, while $v$ is adjacent to $t_1$: it is a star graph with $u$ in the center and $v$ on the branch between $u$ and $t_1$, see Figure~\ref{fig:three-terminal-tree-mincuts}, first panel.  Since the graph
  is a tree, separating the three terminals requires cutting at least two
  edges, and this lower bound is achieved. Hence
  \begin{equation}
    \mathcal{A}(R_1:R_2:R_3)=2.
  \end{equation}

  The elements of $\mathfrak{M}$ are easy to list.  A minimal multiway
  cut labeling $\lambda$ is determined by the pair of colors of the bulk vertices $u$ and $v$
  \begin{equation}
    (a,b):=(\lambda(u),\lambda(v)).
  \end{equation}
  The cut area is equal to $2$ exactly for
  \begin{equation}
    \mathfrak{M}
    =
    \{(1,1),\ (2,1),\ (2,2),\ (3,1),\ (3,3)\},
    \label{eq:tree-example-five-mincuts}
  \end{equation}
  where the notation records $(\lambda(u),\lambda(v))$, see Figure~\ref{fig:three-terminal-tree-mincuts}, the last 5 panels.

  \begin{figure}[htb]
    \centering
    \begin{minipage}{0.25\linewidth}
      \centering
      \includegraphics[width=\linewidth]{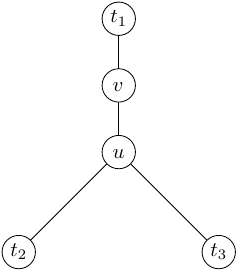}
      \par\smallskip
      \scriptsize (a) The tree
    \end{minipage}
    \hfill
    \begin{minipage}{0.25\linewidth}
      \centering
      \includegraphics[width=\linewidth]{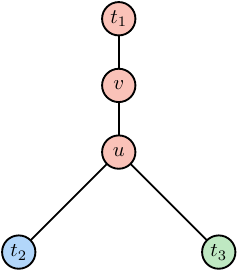}
      \par\smallskip
      \scriptsize (b) $(1,1)$
    \end{minipage}
    \hfill
    \begin{minipage}{0.25\linewidth}
      \centering
      \includegraphics[width=\linewidth]{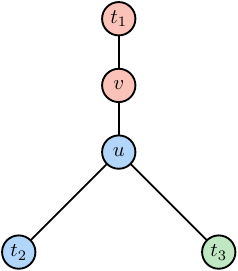}
      \par\smallskip
      \scriptsize (c) $(2,1)$
    \end{minipage}

    \vspace{0.8em}

    \begin{minipage}{0.25\linewidth}
      \centering
      \includegraphics[width=\linewidth]{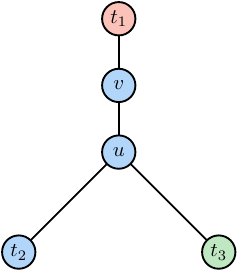}
      \par\smallskip
      \scriptsize (d) $(2,2)$
    \end{minipage}
    \hfill
    \begin{minipage}{0.25\linewidth}
      \centering
      \includegraphics[width=\linewidth]{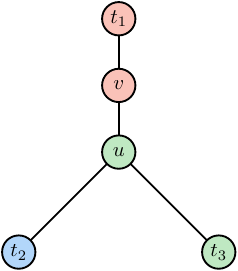}
      \par\smallskip
      \scriptsize (e) $(3,1)$
    \end{minipage}
    \hfill
    \begin{minipage}{0.25\linewidth}
      \centering
      \includegraphics[width=\linewidth]{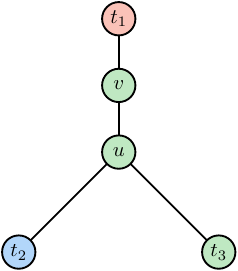}
      \par\smallskip
      \scriptsize (f) $(3,3)$
    \end{minipage}
    \caption{The connected three-terminal tree used in the tree example, followed by its five
    minimal multiway cut labelings.  In panels (b)--(f), the pair
    records the labels $(\lambda(u),\lambda(v))$ of the two bulk
    vertices. Note that the three terminals $\textcolor{pastelone}{t_1}, \textcolor{pasteltwo}{t_2}, \textcolor{pastelthree}{t_3}$ have fixed labels in panels (b)--(f).}
    \label{fig:three-terminal-tree-mincuts}
  \end{figure}

  We now apply Theorem \ref{thm:compatible-cut-geodesic-families}.  Let
  $\Lambda=(\lambda_x)_{x\in X}\in\mathfrak{M}^X$ be a compatible
  family and put
  \begin{equation}
    A_x:=\lambda_x(u),\qquad B_x:=\lambda_x(v).
  \end{equation}
  Thus $(A_x,B_x)$ must be one of the five pairs in
  \eqref{eq:tree-example-five-mincuts} for every $x\in X$.  Conditions
  \ref{item:compatible-vertex} and \ref{item:compatible-edge} give the
  following possibilities.  First, $A_x$ must be independent of $x$:
  \begin{equation}
    A_x= 1,\qquad A_x= 2,\qquad \text{or}\qquad A_x= 3.
  \end{equation}
  Hence $g_u^\Lambda$ is respectively $g_1$, $g_2$, or $g_3$.

  Once this constant value is fixed, condition
  \ref{item:compatible-edge} on the edge $(u,v)$ says exactly that
  $g_v^\Lambda$ lies on the terminal geodesic interval from
  $g_u^\Lambda$ to $g_1$.  Equivalently, $B_x$ is constant on the pairs
  of the matching $\mathcal{P}_{i1}$ when $A_x= i$.  Therefore
  \begin{align}
    |\{g_v^\Lambda:A_x= 1\}|
    &= |[g_1,g_1]|=1,\\
    |\{g_v^\Lambda:A_x= 2\}|
    &= |[g_2,g_1]|=4,\\
    |\{g_v^\Lambda:A_x= 3\}|
    &= |[g_3,g_1]|=4.
  \end{align}
  The counting formula \eqref{eq:minimizer-count-compatible-families}
  thus gives
  \begin{equation}
    |\operatorname{Min}(F)|=1+4+4=9.
  \end{equation}

  The nine minimizers are the following.  The terminal values are fixed
  by admissibility, so we only list the values at $u$ and $v$:
  \begin{equation}
  \begin{array}{c|c}
    g(u) & g(v)\\
    \hline
    g_1 & g_1\\
    g_2 & g_1\\
    g_2 & (00\ 10\ 11\ 01)\\
    g_2 & (00\ 01\ 11\ 10)\\
    g_2 & g_2\\
    g_3 & g_1\\
    g_3 & (00\ 10)\\
    g_3 & (01\ 11)\\
    g_3 & g_3
  \end{array}
  \end{equation}
  with cycles written on the set $X=\{00,10,01,11\}$.

  For instance, the choice
  \begin{equation}
    g(u)=g_3,\qquad g(v)=(00\ 10)
  \end{equation}
  is a minimizer, but $(00\ 10)$ is not one of the terminal
  permutations $g_1,g_2,g_3$.
  This illustrates the point of the
  compatible-family description: each slice $x\mapsto g^\Lambda_v(x)$ is
  associated to a minimal multiway cut, but the minimal cut may depend on
  $x$.

\subsection{Counting minimizers for three-terminal trees}

The tree example above is an instance of a closed formula for weighted
three-terminal trees.  In this subsection we specialize to $q=3$ and $X=\mathbb{Z}_2^2$,
and assume that the terminal regions are singleton vertices $R_i=\{t_i\}$, $i=1,2,3$. Terminal-free branches attached to the minimal subtree spanning
$t_1,t_2,t_3$ never create choices in a minimizer: since all edge
weights are positive, the labels on such a branch must be constant and
equal to the label at the attachment point. Thus only the ``terminal subtree'' matters.

We first record the corresponding weighted one-vertex consequence of
Theorem \ref{thm:1TN_w}.

\begin{corollary}
  \label{cor:weighted-q3-one-vertex}
  Let $a_1,a_2,a_3>0$.  For $h\in \Sym(X)$ consider the quantity $F$ from \eqref{eq:F(g)}
  \begin{equation}
    F(h)=a_1d(h,g_1)+a_2d(h,g_2)+a_3d(h,g_3).
  \end{equation}
  If $a_{\max}=\max(a_1,a_2,a_3)$, then
  \begin{equation}
    \label{eq:weighted-q3-one-vertex-value}
    \min_{h\in \Sym(X)}F(h)=2(a_1+a_2+a_3-a_{\max}),
  \end{equation}
  and the minimizers are exactly $\{g_k \, : \, a_k=a_{\max}\}$.
\end{corollary}

We now consider the free energy to be minimized on an arm of the star graph.
The following lemma uses the term \emph{geodesic subdivision} in the
following sense.  If the path in the graph between two vertices $\sigma$ and $\tau$ consists of $\rho$ edges, then a geodesic subdivision is a sequence
\begin{equation}
  \sigma=\pi_0 \to \pi_1 \to \cdots \to \pi_\rho=\tau
\end{equation}
whose total length is the distance between its endpoints:
\begin{equation}
  \sum_{\ell=1}^{\rho}d(\pi_{\ell-1},\pi_\ell)=d(\sigma,\tau).
\end{equation}
Thus the intermediate labels subdivide a shortest path from $\sigma$ to
$\tau$ in the Cayley graph of the symmetric group $\Sym_n$; some pieces may have length zero, corresponding to no label change across that edge. In what follows we shall consider the case where the permutation $\sigma^{-1}\tau$ consists of a product of disjoint transpositions. The case where $\sigma$ is the identity permutation and $\tau$ is a full cycle in $\Sym_n$ is very important in free probability theory, see \cite{banica2011free} or \cite{Fitter:2024nxn,Hu:2025geh} for some application to random tensor network theory.

\begin{lemma}
  \label{lem:weighted-path-geodesic-count}
  Let $P=(v_0,v_1,\ldots,v_m)$ be a path with positive edge weights, and fix endpoint labels
  \begin{equation}
    h(v_0)=\sigma,\qquad h(v_m)=\tau.
  \end{equation}
  Consider the minimal weight on the path and its multiplicity:
  \begin{equation}
    \mu=\min_{1\le r\le m}w(v_{r-1},v_r),\qquad
    \rho=|\{r \, : \, w(v_{r-1},v_r)=\mu\}|.
  \end{equation}
  Then the minimum possible path energy is
  \begin{equation}
    \mu\,d(\sigma,\tau).
  \end{equation}
  The minimizing labelings are obtained as follows: labels are constant
  across every edge of weight larger than $\mu$, and across the $\rho$
  minimum-weight edges they form a geodesic subdivision
  \begin{equation}
    \sigma=\pi_0 \to \pi_1 \to \cdots \to \pi_\rho=\tau
  \end{equation}
  satisfying
  \begin{equation}
    \sum_{\ell=1}^{\rho}d(\pi_{\ell-1},\pi_\ell)=d(\sigma,\tau).
  \end{equation}
  Let $\mathsf G_\rho(\sigma,\tau)$ denote the number of such geodesic
  subdivisions.  If $\sigma^{-1}\tau$ is a product of $D$ disjoint
  transpositions, then
  \begin{equation}
    \label{eq:path-geodesic-count-involution}
    \mathsf G_\rho(\sigma,\tau)=\rho^D.
  \end{equation}
\end{lemma}
\begin{proof}
  For any labeling of the path,
  \begin{equation}
    \sum_{r=1}^m w(v_{r-1},v_r)d(h(v_{r-1}),h(v_r))
    \ge
    \mu\sum_{r=1}^m d(h(v_{r-1}),h(v_r))
    \ge
    \mu\,d(\sigma,\tau),
  \end{equation}
  by the triangle inequality.  Equality holds precisely when every
  positive distance occurs on an edge of weight $\mu$ and the sequence of
  labels seen on those minimum-weight edges is geodesic from $\sigma$ to
  $\tau$.  This proves the first assertion and the description of
  minimizers.

  If $\sigma^{-1}\tau$ is a product of $D$ disjoint transpositions, then
  a Cayley geodesic from $\sigma$ to $\tau$ is obtained by applying each
  of these $D$ transpositions exactly once, in any order; see also Eq.~\eqref{eq:terminal-geodesic-matching}.  A geodesic
  subdivision across $\rho$ ordered minimum-weight edges is therefore the
  same thing as assigning each of the $D$ transpositions to one of the
  $\rho$ edges on which it is applied.  Since the transpositions are
  disjoint and commute, the intermediate labels are determined by these
  assignments.  Hence there are $\rho^D$ subdivisions.
\end{proof}

We arrive now at the main result of this section, the minimizer count for weighted three-terminal trees.

\begin{theorem}
  \label{thm:q3-tree-minimizer-count}
  Let $T$ be a connected weighted tree with singleton terminal vertices
  $t_1,t_2,t_3$ in the binary three-terminal problem
  $q=3$, $X=\mathbb{Z}_2^2$. Let $m$ be the \emph{median vertex} of the three terminals, i.e.~the unique
  vertex in the intersection $[t_1,t_2]\cap[t_1,t_3]\cap[t_2,t_3]$.

  Let $P_i$ be the path from $m$ to $t_i$; we assume that all the $P_i$'s have positive length, otherwise the tree is degenerate. Define
  \begin{equation}
    \mu_i=\min_{e\in P_i}w(e),\qquad
    \rho_i=|\{e\in P_i \, :\, w(e)=\mu_i\}|,
  \end{equation}
  \begin{equation}
    K_{\max}:=\{k\in\{1,2,3\}:\mu_k=\max(\mu_1,\mu_2,\mu_3)\}.
  \end{equation}
  Then
  \begin{equation}
    \label{eq:q3-tree-count-formula}
    |\operatorname{Min}(F)|
    =
    \sum_{k\in K_{\max}}\prod_{i\ne k}\rho_i^2.
  \end{equation}
\end{theorem}
\begin{proof}
  First remove all branches that do not contain a terminal.  In any
  minimizer such a branch must be constant, because it has no boundary
  condition and all edge weights are positive.  Thus it contributes no
  factor to the count.

  The assumption in the statement guarantees that $m$ is not itself one of the terminals. Fix the label $h:=g(m)$, where $g$ is a minimizer. By Lemma \ref{lem:weighted-path-geodesic-count}, the minimum energy on
  the arm $P_i$, with endpoint labels $h$ and $g_i$, is $\mu_i d(h,g_i)$. Therefore the optimization over the whole tree reduces to the weighted
  one-vertex problem
  \begin{equation}
    \min_{h\in \Sym(X)}\sum_{i=1}^3\mu_i d(h,g_i).
  \end{equation}
  By Corollary \ref{cor:weighted-q3-one-vertex}, the minimizing labels at
  $m$ are exactly $h=g_k$, $k\in K_{\max}$. For such a $k$, the arm $P_k$ has equal endpoint labels and contributes
  one labeling.  For $i\ne k$, the relative permutation $g_k^{-1}g_i$ is
  a product of $2^{q-2}=2$ disjoint transpositions on $X=\mathbb{Z}_2^2$.  Lemma
  \ref{lem:weighted-path-geodesic-count} therefore gives
  \begin{equation}
    \mathsf G_{\rho_i}(g_k,g_i)=\rho_i^2
  \end{equation}
  labelings on $P_i$.  Multiplying over the three independent arms and
  summing over $k\in K_{\max}$ gives
  \eqref{eq:q3-tree-count-formula}.
\end{proof}

For the unit-weight tree example in Figure \ref{fig:three-terminal-tree-mincuts}, the median is $u$ and
\begin{equation}
  (\rho_1,\rho_2,\rho_3)=(2,1,1),\qquad
  K_{\max}=\{1,2,3\}.
\end{equation}
The formula gives
\begin{equation}
  |\operatorname{Min}(F)|
  =
  1^2\cdot 1^2
  +2^2\cdot 1^2
  +2^2\cdot 1^2
  =9,
\end{equation}
recovering the count computed above.

As another simple illustration, for a weighted three-terminal star with a unique bulk vertex $b$ and edge weights
\begin{equation}
  w_{b t_1}=w_{b t_2}=1,\qquad w_{b t_3}=\varepsilon,\qquad 0<\varepsilon<1,
\end{equation}
one has
\begin{equation}
  K_{\max}=\{1,2\},\qquad
  \rho_1=\rho_2=\rho_3=1,
\end{equation}
so \eqref{eq:q3-tree-count-formula} gives
\(|\operatorname{Min}(F)|=2\).

In the case of unit-weight three-terminal trees, we have the following minimizer count; the proof follows directly from Theorem \ref{thm:q3-tree-minimizer-count} and is left to the reader.

\begin{corollary}
  \label{cor:unit-weight-q3-tree-minimizer-count}
  Let $T$ be a connected three-terminal tree in the binary
  three-terminal problem $q=3$, $X=\mathbb{Z}_2^2$, with singleton terminals $t_1,t_2,t_3$, and assume $w(e)=1$ for every edge $e$. Let $m$ be the median vertex of the three terminals, and let
  \begin{equation}
    \ell_i=|E([m,t_i])|\qquad (i=1,2,3)
  \end{equation}
  be the combinatorial lengths of the three terminal arms.

  Assuming $m$ is not one of the terminals, we have
  \begin{equation}
    \label{eq:unit-weight-q3-tree-count}
    |\operatorname{Min}(F)|
    =
    \ell_2^2\ell_3^2
    +\ell_1^2\ell_3^2
    +\ell_1^2\ell_2^2.
  \end{equation}
\end{corollary}

\subsection{The triangle network}

We next compute a small connected non-tree example with symmetric
weights: a \emph{triangle graph}.  The result exhibits a sharp
transition between a degenerate constant regime and a unique aligned
regime.

\begin{proposition}
Fix $q=3$ and $X=\mathbb{Z}_2^2$.  Let $G$ be the triangle network in
Figure~\ref{fig:triangle-network-mincuts}, with singleton terminal
vertices $t_1,t_2,t_3$ and bulk vertices $u_1,u_2,u_3$.
The edge set is
\begin{equation}
\{t_i,u_i\}\quad (1\le i\le 3),
\qquad
\{u_i,u_j\}\quad (1\le i<j\le 3),
\end{equation}
with weights $w_{t_i u_i}=a$ and $w_{u_i u_j}=b$ for $1\le i<j\le 3$,
where $a,b>0$.
Then the minimizers of the RTN energy minimization problem defined
in \eqref{eq:RTN_F(g)} are as follows:

\begin{enumerate}
\item If $0<a<\frac32 b$, then $\min F=4a$, and there are exactly
three minimizers:
\begin{equation}\label{eq:triangle-small-spoke-minimizers}
g_{u_1}=g_{u_2}=g_{u_3}=g_k
\qquad (k=1,2,3).
\end{equation}

\item If $a=\frac32 b$, then $\min F=4a=6b$, and there are exactly four
minimizers: the three constant minimizers in
\eqref{eq:triangle-small-spoke-minimizers}, together with the aligned
minimizer
\begin{equation}\label{eq:triangle-aligned-minimizer}
g_{u_i}=g_i
\qquad (i=1,2,3).
\end{equation}

\item If $a>\frac32 b$, then $\min F=6b$, and the aligned labeling
\eqref{eq:triangle-aligned-minimizer} is the unique minimizer.
\end{enumerate}

Thus the symmetric triangle network has more than one minimizer
precisely when $a\le \frac32 b$.
\end{proposition}

\begin{figure}[htb]
  \centering
  \begin{minipage}{0.28\linewidth}
    \centering
    \includegraphics[width=\linewidth]{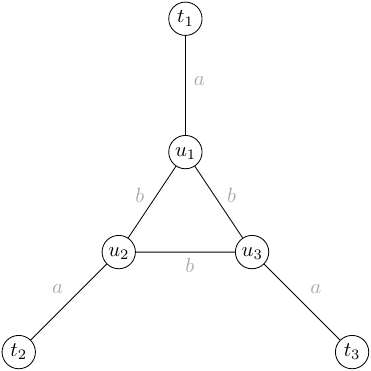}
    \par\smallskip
    \scriptsize (a) The triangle network
  \end{minipage}
  \hfill
  \begin{minipage}{0.28\linewidth}
    \centering
    \includegraphics[width=\linewidth]{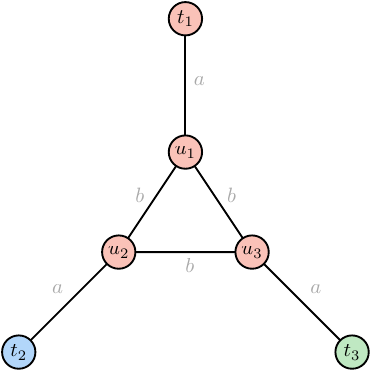}
    \par\smallskip
    \scriptsize (b) $C_1$
  \end{minipage}
  \hfill
  \begin{minipage}{0.28\linewidth}
    \centering
    \includegraphics[width=\linewidth]{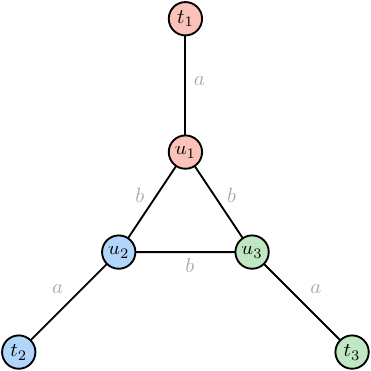}
    \par\smallskip
    \scriptsize (c) $I$
  \end{minipage}
  \caption{The symmetric triangle network, followed by representative
  minimum multiway cut labelings.  Panel (b) shows the constant labeling
  $C_1$, while panel (c) shows the aligned labeling $I$. The three terminals $\textcolor{pastelone}{t_1}, \textcolor{pasteltwo}{t_2}, \textcolor{pastelthree}{t_3}$ have fixed labels in the last two panels.}
  \label{fig:triangle-network-mincuts}
\end{figure}

\begin{proof}
We first classify the minimum multiway partitions.  A terminal-valued
partition is determined by labels $\ell_i\in\{1,2,3\}$ assigned to the
three bulk vertices.  Its cut area is
\begin{equation}\label{eq:triangle-network-general-cut-area}
C_{a,b}(\ell_1,\ell_2,\ell_3)
=
a\sum_{i=1}^3\mathbf 1_{\ell_i\neq i}
+
b\sum_{1\le i<j\le 3}\mathbf 1_{\ell_i\neq \ell_j}.
\end{equation}
There are two evident symmetric competitors.  The aligned partition
$\ell_i=i$ for $1\le i\le 3$ has area $3b$.  The three constant
partitions $\ell_1=\ell_2=\ell_3=k$ for $k=1,2,3$ have area $2a$.
Every other labeling of the bulk vertices is neither constant nor
aligned.  Hence it has at least one wrong terminal spoke and at least
two disagreeing bulk-triangle edges, so by
\eqref{eq:triangle-network-general-cut-area} its area is at least
$a+2b$.  If $a\le 3b/2$, then $a+2b>2a$, and if $a\ge 3b/2$, then
$a+2b>3b$.
Therefore the minimum cut labelings are exactly:
\[
\begin{array}{c|c|c}
\text{regime} & \mathcal A & \text{minimum cut labelings}\\
\hline
0<a<\frac32 b & 2a & C_1,C_2,C_3\\
a=\frac32 b & 2a=3b & C_1,C_2,C_3,I\\
a>\frac32 b & 3b & I,
\end{array}
\]
where $C_k(u_1)=C_k(u_2)=C_k(u_3)=k$ and $I(u_i)=i$. By Theorem~\ref{thm:RTN} applied with $q=3$, $\min F=2\mathcal A$, which gives the three displayed values of $\min F$.

It remains to count all permutation minimizers, not only the
terminal-valued ones.  If $a>3b/2$, the minimum multiway partition is
unique, so Corollaries~\ref{cor:unique-mincut} or \ref{cor:unique-mincut-from-characterization} force every minimizer to be terminal-valued on that partition.  This gives the
unique aligned minimizer.

If $0<a<3b/2$, then every slice of a minimizer must use one of the
constant minimum cut labelings $C_1,C_2,C_3$.  Thus for each
$x\in X$ all three bulk vertices have the same slice value,
$g_{u_1}(x)=g_{u_2}(x)=g_{u_3}(x)$.  Hence
$g_{u_1}=g_{u_2}=g_{u_3}=h$ for some permutation
$h\in \Sym(X)$.  The energy then reduces to
$F=a\sum_{i=1}^3 d(h,g_i)$.  Corollary~\ref{cor:weighted-q3-one-vertex},
applied with $a_1=a_2=a_3=1$, gives
$\sum_{i=1}^3 d(h,g_i)\ge 4$, with equality exactly for
$h\in\{g_1,g_2,g_3\}$.
Thus the only minimizers in this regime are the three constant
terminal-valued labelings in
\eqref{eq:triangle-small-spoke-minimizers}.

Finally assume $a=\frac32 b$.
The minimum cut labelings are $C_1,C_2,C_3,I$.  Applying the compatible
cut-geodesic formula of Theorem
\ref{thm:compatible-cut-geodesic-families} to this four-element set
gives only the four constant families
\begin{equation}\label{eq:triangle-critical-compatible-families}
\lambda_x= C_1,\qquad
\lambda_x= C_2,\qquad
\lambda_x= C_3,\qquad
\lambda_x= I.
\end{equation}
Indeed, let $\Lambda=(\lambda_x)_{x\in X}$ be a compatible family.  We
claim that if one slice is one of the constant labelings, then all
slices are the same constant labeling.  Suppose
$\lambda_x=C_k$.  For each $i\ne k$, condition
\ref{item:compatible-edge} on the spoke $t_i u_i$ gives
\begin{equation}
  \lambda_{x+e_i+e_k}(u_i)=k.
\end{equation}
Among the four minimum cut labelings $C_1,C_2,C_3,I$, the only one
which labels $u_i$ by $k$ with $i\ne k$ is $C_k$ itself, since
$I(u_i)=i$.  Hence
\begin{equation}
  \lambda_{x+e_i+e_k}=C_k
  \qquad (i\ne k).
\end{equation}
The two vectors $e_i+e_k$ with $i\ne k$ are distinct nonzero elements
of $X=\mathbb Z_2^2$, and hence generate $X$.  Iterating the preceding
implication therefore forces $\lambda_y=C_k$ for every $y\in X$.
Thus no compatible family can mix a constant slice with any other type,
or mix two different constant types.  If no slice is one of
$C_1,C_2,C_3$, then every slice is $I$, so the family is
$\lambda_x= I$.

Conversely, each of the four constant-in-$x$ families in
\eqref{eq:triangle-critical-compatible-families} is compatible:
condition \ref{item:compatible-vertex} holds because every vertex map
is a fixed translation $x\mapsto x+e_j$, and condition
\ref{item:compatible-edge} is immediate because the labels do not
depend on $x$.  They give exactly the three constant minimizers and the
aligned minimizer.  This proves the minimizer count.
\end{proof}

Non-terminal-valued minimizers arise when the equality theorem permits
different minimizing partitions on different slices of $X$, and
the resulting compatibility constraints still leave enough freedom to
mix those slices coherently. This slice-mixing mechanism, and the
structural conditions that prevent or allow it, will be explored in
depth in the next subsection.

\subsection{A sufficient condition for terminal-valued minimizers}

We now record a simple condition under which all the minimizers are terminal-valued, and thus in bijection with the set $\mathfrak M$ of minimal multiway cuts. Recall that $n=2$, $G$ is connected and that all edge weights are strictly positive.  Let us introduce the set of oriented graph edges
\begin{equation}
  \overrightarrow E
  :=
  \{(u,v):(u,v)\in E\}.
\end{equation}
For $\lambda\in\mathfrak{M}$ and distinct terminal labels $i\ne j$, define the oriented
$i \to j$ cut edges of the minimum cut labeling $\lambda$ by
\begin{equation}\label{eq:oriented-crossing-pattern}
  E_{ij}(\lambda)
  :=
  \{(u,v)\in\overrightarrow E:
  \lambda(u)=i,\ \lambda(v)=j\}.
\end{equation}
Thus $E_{ij}(\lambda)$ remembers not only which edges cross between the ``$i$'' and ``$j$'' regions of the cut labeling, but also on which side of the edge the two terminal labels lie.

We say that a minimizer $g:V\to \Sym(X)$ is \emph{terminal-valued} if there
exists some $\lambda\in\mathfrak{M}$ such that
\begin{equation}
  g(v)=g_{\lambda(v)},\quad v\in V.
\end{equation}

\begin{proposition}\label{prop:oriented-edge-pattern-rigidity}
Assume that, for every $i=1,\ldots,q-1$, the map
\begin{equation}\label{eq:star-oriented-pattern-injectivity}
  \mathfrak{M}\to 2^{\overrightarrow E},
  \qquad
  \lambda\mapsto E_{iq}(\lambda)
\end{equation}
is \emph{injective}.  Then every minimizer of $F$ is terminal-valued.
\end{proposition}

\begin{proof}
Let $g\in\operatorname{Min}(F)$ be a minimizer.  By Theorem~\ref{thm:compatible-cut-geodesic-families}, there is a compatible family $\Lambda=(\lambda_x)_{x\in X}\in\mathfrak{M}^X$ such that, for all $v \in V$ and $x \in X$,
\begin{equation}\label{eq:rigidity-compatible-family-representation}
  g(v)(x)=g_{\lambda_x(v)}(x)=x+e_{\lambda_x(v)}.
\end{equation}
Our goal is to show that $\lambda_x$ is independent of $x$.

First fix distinct $i,j\in\{1,\ldots,q\}$.  We claim that the compatibility equations force
\begin{equation}\label{eq:oriented-pattern-slice-translation}
  E_{ij}(\lambda_x)=E_{ij}(\lambda_{x+e_i+e_j})
  \qquad (x\in X).
\end{equation}
Condition~\ref{item:compatible-edge} may be applied to either
orientation of an undirected edge: reversing the endpoints only swaps
$i$ and $j$, while the translation $x+e_i+e_j$ is unchanged.  Now let
$y=x+e_i+e_j$.  If $(u,v)\in E_{ij}(\lambda_x)$, then $\lambda_x(u)=i$ and $\lambda_x(v)=j$. Applying condition~\ref{item:compatible-edge} to the edge $(u,v)$ gives
\begin{equation}
  \lambda_y(u)=i,\qquad \lambda_y(v)=j,
\end{equation}
so $(u,v)\in E_{ij}(\lambda_y)$.  This proves one inclusion in
\eqref{eq:oriented-pattern-slice-translation}.  The reverse inclusion
follows by the same argument with $x$ replaced by $y$, since over $\mathbb{Z}_2$ one has $y+e_i+e_j=x$; this proves the claim in \eqref{eq:oriented-pattern-slice-translation}.

Taking $j=q$ and using $e_q=0$, we obtain, for all $i=1,\ldots,q-1$ and $x\in X$,
\begin{equation}
  E_{iq}(\lambda_x)=E_{iq}(\lambda_{x+e_i}).
\end{equation}
By the injectivity assumption \eqref{eq:star-oriented-pattern-injectivity}, this implies that $\lambda_x=\lambda_{x+e_i}$. Since the vectors $e_1,\ldots,e_{q-1}$ generate $X=\mathbb{Z}_2^{q-1}$, these equalities force all slices to use the same minimum cut labeling: there is a fixed $\lambda\in\mathfrak{M}$ with $\lambda_x=\lambda$ for every $x\in X$. Substituting this into
\eqref{eq:rigidity-compatible-family-representation} gives
\begin{equation}
  g(v)(x)=g_{\lambda(v)}(x)
  \qquad  \forall v\in V,\ x\in X,
\end{equation}
and therefore $g(v)=g_{\lambda(v)}$ for every $v\in V$.  Thus $g$ is
terminal-valued.
\end{proof}

\begin{remark}
A careful analysis of the proof shows that the conclusion still holds
if one replaces the assumption by the following
condition.  Let
\begin{equation}
  \mathcal P\subseteq\{(i,j):1\le i<j\le q\}
\end{equation}
be a set of terminal pairs such that
\begin{equation}
  \{e_i+e_j:\ (i,j)\in\mathcal P\}
\end{equation}
spans $X=\mathbb Z_2^{q-1}$.  Assume that, for every
$(i,j)\in\mathcal P$, the coordinate map
\begin{equation}
  \mathfrak{M}\to
  2^{\overrightarrow E},
  \qquad
  \lambda\mapsto E_{ij}(\lambda)
\end{equation}
is injective.  Then all minimizers are terminal-valued.
\end{remark}

We consider now two examples, the symmetric triangle network from Figure~\ref{fig:triangle-network-mincuts} and the three-terminal tree from Figure~\ref{fig:three-terminal-tree-mincuts}, to illustrate the criterion in Proposition~\ref{prop:oriented-edge-pattern-rigidity}.

\begin{example}
Consider the triangle network of Figure~\ref{fig:triangle-network-mincuts}.
For $q=3$, the criterion in
Proposition~\ref{prop:oriented-edge-pattern-rigidity} requires
injectivity of the two maps
\begin{equation}
  \lambda\mapsto E_{13}(\lambda),
  \qquad
  \lambda\mapsto E_{23}(\lambda).
\end{equation}
It is enough to check this at the critical ratio $a=\frac32b$, where
the set of minimum cut labelings is largest:
\begin{equation}
  \mathfrak{M}=\{C_1,C_2,C_3,I\}.
\end{equation}
Here $C_k$ labels all three bulk vertices by $k$, while
$I(u_i)=i$.  A direct computation from
\eqref{eq:oriented-crossing-pattern} gives
\begin{equation}
\begin{array}{c|c|c}
\lambda & E_{13}(\lambda) & E_{23}(\lambda)\\
\hline
C_1 & \{(u_3,t_3)\} & \varnothing\\
C_2 & \varnothing & \{(u_3,t_3)\}\\
C_3 & \{(t_1,u_1)\} & \{(t_2,u_2)\}\\
I   & \{(u_1,u_3)\} & \{(u_2,u_3)\}.
\end{array}
\end{equation}
The entries in the $E_{13}$ column are pairwise distinct, and the
entries in the $E_{23}$ column are pairwise distinct.  Hence both maps
$\mathfrak{M}\to 2^{\overrightarrow E}$ are injective at the critical
ratio.  In the regimes $0<a<\frac32b$ and $a>\frac32b$ the set
$\mathfrak{M}$ is respectively the subset $\{C_1,C_2,C_3\}$ and the
singleton $\{I\}$, so injectivity follows by restriction.  Therefore
Proposition~\ref{prop:oriented-edge-pattern-rigidity} applies in every
regime of the symmetric triangle network, giving another proof that all
its minimizers are terminal-valued.
\end{example}

\begin{example}[the three-terminal tree]
Now consider the connected tree of Figure~\ref{fig:three-terminal-tree-mincuts},
with unit edge weights and edges
\begin{equation}
  \{t_3,u\},\qquad \{t_2,u\},\qquad \{u,v\},\qquad \{v,t_1\}.
\end{equation}
As in \eqref{eq:tree-example-five-mincuts}, the minimum cut labelings
are
\begin{equation}
  \mathfrak{M}
  =
  \{\lambda_{11},\lambda_{21},\lambda_{22},\lambda_{31},\lambda_{33}\},
\end{equation}
where $\lambda_{ab}(u)=a$ and $\lambda_{ab}(v)=b$, while
$\lambda_{ab}(t_i)=i$ on the terminals.  The two star-pair oriented
crossing patterns are
\begin{equation}
\begin{array}{c|c|c}
\lambda & E_{13}(\lambda) & E_{23}(\lambda)\\
\hline
\lambda_{11} & \{(u,t_3)\} & \varnothing\\
\lambda_{21} & \varnothing & \{(u,t_3)\}\\
\lambda_{22} & \varnothing & \{(u,t_3)\}\\
\lambda_{31} & \{(v,u)\} & \{(t_2,u)\}\\
\lambda_{33} & \{(t_1,v)\} & \{(t_2,u)\}.
\end{array}
\end{equation}
Thus the map $\lambda\mapsto E_{13}(\lambda)$ is not injective, since
\begin{equation}
  E_{13}(\lambda_{21})=E_{13}(\lambda_{22})=\varnothing
  \qquad\text{but}\qquad
  \lambda_{21}\ne\lambda_{22}.
\end{equation}
Likewise, $\lambda\mapsto E_{23}(\lambda)$ is not injective, for
example
\begin{equation}
  E_{23}(\lambda_{31})=E_{23}(\lambda_{33})=\{(t_2,u)\}
  \qquad\text{but}\qquad
  \lambda_{31}\ne\lambda_{33}.
\end{equation}
Hence, the oriented-edge-pattern rigidity criterion fails for this tree.
This failure is consistent with the explicit non-terminal-valued
minimizers found above: distinct minimum cut labelings can have the
same oriented crossing pattern in a terminal-pair direction, so the
compatibility equations do not force all slices $\lambda_x$ to agree.
\end{example}

\begin{remark}
The unique-minimum-partition case is the special case $|\mathfrak{M}|=1$, for which the injectivity condition is automatic. Proposition~\ref{prop:oriented-edge-pattern-rigidity} goes beyond the uniqueness of the minimum multiway cut: it allows several minimum cut labelings, provided their oriented crossing patterns against the terminal $q$ already distinguish them.  Under that condition, the edge compatibility relations synchronize the minimizing slices across $X$, so no non-terminal-valued permutation minimizer can be formed by mixing different minimum cuts at different slice points.
\end{remark}

\section{A conjecture for higher \texorpdfstring{$n$}{n}}\label{sec:n>2}

We now return to the R\'enyi index $n>2$; actually, most of the results in this section also hold for $n=2$, but we focus on the strict inequality case since the binary case has been dealt with in detail in the previous sections.  The terminal permutations
are still the twist operators from Definition~\ref{def:mutli_ent}, but
it is useful to spell out the notation in the form used below.  Put
$X=\mathbb{Z}_n^{q-1}$, let $e_1,\ldots,e_{q-1}$ be the standard
generators of $X$, and put $e_q=0$.  For $a\in X$, write
$\tau_a:X\to X$, $\tau_a(x)=x+a$, for translation by $a$.  The $q$
terminal permutations are $g_i=\tau_{e_i}$ for $1\le i\le q$, and
$g_q=\id$.  Equivalently, for $i<q$, $g_i$ translates the $i$-th
coordinate by one; in this way, $g_i$ is a product of $n^{q-2}$ disjoint cycles of length $n$.  For example, when $q=3$ and $n=5$, one has (see also Fig.~\ref{fig:Z52}, left panel)
$X=\mathbb Z_5^2$ and
\begin{equation}
  g_1(a,b)=(a+1,b),\qquad
  g_2(a,b)=(a,b+1),\qquad
  g_3(a,b)=(a,b),
\end{equation}
with all coordinates understood modulo $5$.  In cycle notation this is
\begin{equation}
\begin{aligned}
g_1
&=\prod_{b\in\mathbb Z_5}
  \bigl((0,b)\ (1,b)\ (2,b)\ (3,b)\ (4,b)\bigr),\\
g_2
&=\prod_{a\in\mathbb Z_5}
  \bigl((a,0)\ (a,1)\ (a,2)\ (a,3)\ (a,4)\bigr),\\
g_3&=\id .
\end{aligned}
\end{equation}

\begin{figure}[htb]
  \centering
  \includegraphics[width=0.4\linewidth]{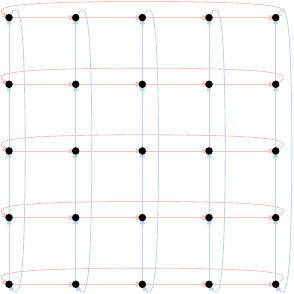}
  \hfill
  \includegraphics[width=0.4\linewidth]{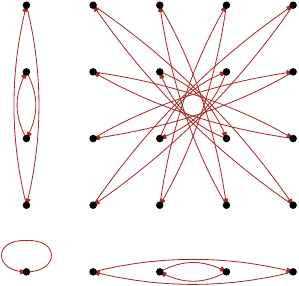}
  \caption{For the example $X=\mathbb Z_5^2$, the left panel shows the
  terminal permutations: $g_1$ translates each horizontal copy of
  $\mathbb Z_5$ by one, $g_2$ translates each vertical copy of
  $\mathbb Z_5$ by one, and $g_3=\id$ (not illustrated).  The right
  panel shows the conjectured minimizer for the single-tensor graph,
  namely the inversion map $\pi(x)=-x$.}
  \label{fig:Z52}
\end{figure}

For $i\ne j$, the relative terminal permutation
$g_i^{-1}g_j=\tau_{e_j-e_i}$ is a product of $n^{q-2}$ disjoint cycles
of length $n$.  Hence $d(g_i,g_j)=n^{q-2}(n-1)$ for $i\ne j$. As in the binary case, there is no common geodesic element, belonging to the intersection of all intervals $[g_i,g_j]$. Actually the intersection of two arbitrary geodesic intervals is as small as it could be.

\begin{proposition}
\label{prop:geodesic}
For distinct unordered pairs of terminals $\{i,j\}$ and $\{k,\ell\}$, one has
\begin{equation}
[g_i,g_j]\cap[g_k,g_\ell]
=
\begin{cases}
\{g_m\}, & \{i,j\}\cap\{k,\ell\}=\{m\},\\
\varnothing, & \{i,j\}\cap\{k,\ell\}=\varnothing.
\end{cases}
\label{eq:pairwise-geodesic-interval-intersection}
\end{equation}
Consequently the relative interiors
$[g_i,g_j]\setminus\{g_i,g_j\}$ are pairwise disjoint.
\end{proposition}
\begin{proof}
  Fix $a\ne b$.  Recall that, for $v\in X$, $\tau_v$ denotes the
  translation permutation $\tau_v(x)=x+v$.  By left-invariance of the
  Cayley metric, if
  $\alpha\in[g_a,g_b]$, then
  \begin{equation}
    \beta:=g_a^{-1}\alpha\in[\id,g_a^{-1}g_b]
    =
    [\id,\tau_{e_b-e_a}].
  \end{equation}
  The cycles of $\tau_{e_b-e_a}$ are exactly the affine lines
  \begin{equation}
    x+\langle e_b-e_a\rangle,\qquad x\in X.
  \end{equation}
  The geodesic fact above therefore implies that $\beta$ preserves each
  such affine line.  Equivalently, for every $x\in X$,
  \begin{equation}
    \label{eq:geodesic-affine-line-constraint}
    \alpha(x)-x
    \in e_a+\langle e_b-e_a\rangle
    =
    \{(1-t)e_a+te_b:t\in\mathbb Z_n\}.
  \end{equation}

  Now suppose first that the two unordered pairs share exactly one
  terminal.  Relabeling the three indices, it is enough to consider
  $\{i,j\}$ and $\{i,k\}$, with $i,j,k$ distinct.  If
  $\alpha\in[g_i,g_j]\cap[g_i,g_k]$, then
  \eqref{eq:geodesic-affine-line-constraint} gives, for every $x$,
  \begin{equation}
    \alpha(x)-x\in
    \bigl(e_i+\langle e_j-e_i\rangle\bigr)
    \cap
    \bigl(e_i+\langle e_k-e_i\rangle\bigr).
  \end{equation}
  This intersection is the single point $\{e_i\}$.  To see this, write
  \begin{equation}
    (1-t)e_i+te_j=(1-s)e_i+se_k,\qquad s,t\in\mathbb Z_n.
  \end{equation}
  Since the vectors $e_1,\ldots,e_{q-1},e_q=0$ are the vertices of the
  standard simplex, comparing the coordinates indexed by
  $\{i,j,k\}\setminus\{q\}$ gives $s=t=0$.  Thus
  $\alpha(x)=x+e_i=g_i(x)$ for every $x$, and hence
  $\alpha=g_i$.  This proves
  $[g_i,g_j]\cap[g_i,g_k]=\{g_i\}$, and the same argument gives the
  stated singleton $\{g_m\}$ for any common terminal $m$.

  Finally suppose that $\{i,j\}\cap\{k,\ell\}=\varnothing$.  If
  $\alpha$ belonged to both intervals, then for every $x$ the
  displacement $\alpha(x)-x$ would lie in
  \begin{equation}
    \bigl(e_i+\langle e_j-e_i\rangle\bigr)
    \cap
    \bigl(e_k+\langle e_\ell-e_k\rangle\bigr).
  \end{equation}
  This intersection is empty.  Indeed an equality
  \begin{equation}
    (1-t)e_i+te_j=(1-s)e_k+se_\ell
  \end{equation}
  would give two affine combinations of the simplex vertices with total
  coefficient $1$.  Comparing the coordinates of
  $e_1,\ldots,e_{q-1}$, and then using the total coefficient, gives
  uniqueness of these affine coordinates.  This is impossible because
  the two unordered pairs of vertices are disjoint.  Hence the two
  intervals are disjoint.
\end{proof}

\subsection{The single tensor minimizer}
Recall that to find the saddle of a single random tensor, we need to optimize the following quantity
\begin{equation}
  \min_{g\in \Sym(X)} F(g) = \min_{g\in \Sym(X)} \sum_{i=1}^q d(g, g_i)
\end{equation}
where we set all edge weights to be equal to 1 and we consider $X=\mathbb{Z}_n^{q-1}$ for any $q$ and $n>2$ in this section.

If the multiway cut conjecture holds as in Section~\ref{sec:n=2}, then $\{g_1, \dots, g_q\}$ is expected to be the set of minimizers for the free energy and one can obtain the following values for the free energy
\begin{equation}
  F(g_1) =  \sum_{i=1}^qd(g_1, g_i)= (q-1)n^{q-2}(n-1)
\end{equation}
However, the above value cannot always be the minimum value of the free energy for $n>2$. We propose the following counterexample.

\begin{definition}
  Let the reflection permutation $\pi\in \Sym(X)$ (see Fig.~\ref{fig:Z52} right panel) be defined as
  \begin{equation}
    \pi(x) = -x \quad \forall x\in X
  \end{equation}
  The Cayley distances from $\pi$ to the terminal permutations depend on the parity of $n$. If $n$ is odd,
  \begin{equation}
    d(\pi,g_i)=\frac{n^{q-1}-1}{2},\qquad 1\leq i \leq q.
  \end{equation}
  If $n$ is even,
  \begin{equation}
    d(\pi,g_i)=\frac{n^{q-1}}{2}\quad \text{ for } 1\leq i< q, \qquad \text{ and }
    \qquad
    d(\pi,\id)=\frac{n^{q-1}-2^{q-1}}{2}.
  \end{equation}
  Therefore
  \begin{equation}
    F(\pi) = \sum_{i=1}^qd(\pi, g_i) = \left\{\begin{array}{ll}
      \dfrac{q(n^{q-1}-1)}{2} & \text{if $n$ is odd} \vspace{1em}\\
      \dfrac{qn^{q-1}-2^{q-1}}{2} & \text{if $n$ is even}.
    \end{array}\right.
  \end{equation}
\end{definition}
\begin{remark}
  We emphasize that $F(\pi)\leq F(g_1)$ with the equality holding only when $q=2$ (the bipartite case, corresponding to the R\'enyi entropy, discussed extensively in the literature), $n=2$ (discussed in the previous sections), or $(q, n)=(3, 3)$ (can be shown by enumerating all possible permutations). Hence, the multiway cut conjecture is false for the other choices of $(q, n)$.
\end{remark}

In fact, one can obtain the reflection permutation by solving for a set of conditions that force the bound suggested in Lemma \ref{lem:m(g)} to be saturated.
\begin{lemma}
  We have $d(g, g_i) = \frac{1}{2}m(gg_i^{-1})$ for all $i=1, \dots, q$ if and only if
  \begin{equation}
    g=\left(\prod_{i=1}^{q-1}g_i^{k_i}\right)\pi \quad \text{for some integers $k_i\in \mathbb{Z}_n$}.
  \end{equation}
\end{lemma}
\begin{proof}
  Using Lemma \ref{lem:m(g)}, we find that $d(g, g_i)=d(gg_i^{-1}, \id)=\frac{1}{2}m(gg_i^{-1})$ for all $i$ if and only if
  \begin{equation}\label{eq:reflection decomposition}
    (gg_1^{-1})^2 = \dots = (gg_{q-1}^{-1})^2 = g^2 = \id
  \end{equation}
  where we use $g_q=\id$. Picking a particular $i\in\{1, \dots, q-1\}$, then
  \begin{equation}
    \begin{split}
      (gg_i^{-1})^2=\id, \quad g^2=\id \iff & gg_i^{-1}g^{-1} = g_i, \quad g=g^{-1} \\
      \iff & g(x-e_i) = g(x)+e_i \quad \forall x\in X, \quad g=g^{-1}
    \end{split}
  \end{equation}
  Note that from above, the $j$-th ($j\neq i$) component of $g(x)$ is the same as that of $g(x-e_i)$. Hence, the $j$-th component of $g(x)$ is independent from the $i$-th component of $x$ and the permutation $g$ which satisfies Eq.(\ref{eq:reflection decomposition}) can be decomposed as follows:
  \begin{equation}
    g(x_1, \dots, x_{q-1}) = (\sigma_1(x_1), \dots, \sigma_{q-1}(x_{q-1})) \quad \text{where $\sigma_i\in \Sym_n$ for all $i=1, \dots, q-1$}
  \end{equation}
  Eq.(\ref{eq:reflection decomposition}) is then equivalent to
  \begin{equation}
    \sigma_i^{2} = (\sigma_i\tau_n^{-1})^2=\id\quad \forall i=1, \dots, q-1
  \end{equation}
  where $\tau_n=(1\dots n)$. Observe that
  \begin{equation}
    \begin{split}
      \sigma_i^{2} = (\sigma_i\tau_n^{-1})^2=\id \iff & \sigma_i = \sigma_i^{-1}, \quad\sigma_i\tau_n^{-1}\sigma_i^{-1}=\tau_n \\
      \iff & \exists k_i\in \mathbb{Z}_n: \sigma_i(x_i) = k_i - x_i \quad \forall x_i\in\mathbb{Z}_n \\
      \iff & \exists k_i\in \mathbb{Z}_n: \sigma_i(x_i) = \tau_n^{k_i}\pi(x_i) \\
      \iff & \sigma_i\in\{\tau_n^{k_i}\pi: k_i\in \mathbb{Z}_n\}
    \end{split}
  \end{equation}
  Hence, it follows that
  \begin{equation}
    g=\left(\prod_{i=1}^{q-1}g_i^{k_i}\right)\pi \quad \text{for some integers $k_i\in\mathbb{Z}_n$}
  \end{equation}
  since $g_i^{k_i}(x_1, \dots, x_i, \dots, x_{q-1}) = (x_1, \dots, \tau_n^{k_i}(x_i), \dots, x_{q-1})$.
\end{proof}

Among the set of solutions obtained above, the free energy may take different values but is always minimized when $g=\pi$. We have pointed out that the free energy $F(\pi)$ cannot be greater than $F(g_1)$, where the latter is the lower bound of the free energy suggested by the multiway cut conjecture, and, recalling Sec.~\ref{sec:n=2 single random tensor}, we have shown that a minimizer for the free energy when $n=2$ is obtained if and only if the lower bound in Lemma \ref{lem:m(g)} is always saturated. Hence, it seems reasonable to say that:

\begin{conjecture}\label{conjecture:higher n}
  The reflection permutation $\pi$ gives the true minimum for all $(q, n)$, i.e.
  \begin{equation}
    \min_{g\in \Sym(X)} F(g) = F(\pi)
  \end{equation}
\end{conjecture}

The structure of the minimizer from the conjecture above became apparent to us after running the \emph{SkyDiscover EvoX} optimization pipeline for algorithmic discovery \cite{skydiscover2026}.
As a consistency check, we have also numerically verified that Conjecture \ref{conjecture:higher n} is true for small $(q,n)$ up to $(q,n)=(3, 4)$.

\subsection{Implication for RTNs}
In this subsection, we touch on the implication of the existence of the replica-symmetry-breaking (RSB) element $\pi$ on more general RTNs, along the lines of Ref. \cite{Akella:2026bci}.
We consider RTNs defined on an isometric tiling on a two-dimensional Riemannian manifold $\mathcal{M}$ (with boundary). We will also assume that the tiling is dense enough so that one can approximate the area of minimal cuts by lengths of the corresponding geodesics in $\mathcal{M}$.
The most interesting case is perhaps when one takes $\mathcal{M}$ to be the hyperbolic Poincar\'e disk (e.g. Figure \ref{fig:TN_cuts_intro}), under which the RTN serves as a toy model for the standard AdS$_3$/CFT$_2$ dictionary in the vacuum setting.
However we stress that our result in this section is valid for any two dimensional manifold $\mathcal{M}$, and we expect that it can be readily generalized to higher dimensions.

At first glance, the failure of the multiway cut conjecture for a single random tensor does not necessarily imply the failure of the multiway cut conjecture in arbitrary RTNs.
However, we will see that the existence of the element $\pi$ allows one to construct solutions with lower energy than the minimal multiway cut configuration, provided that the R\'enyi index $n$ is large enough and the multiway cut has non-trivial intersections.

\begin{figure}
    \centering
    \includegraphics[width=0.75\linewidth]{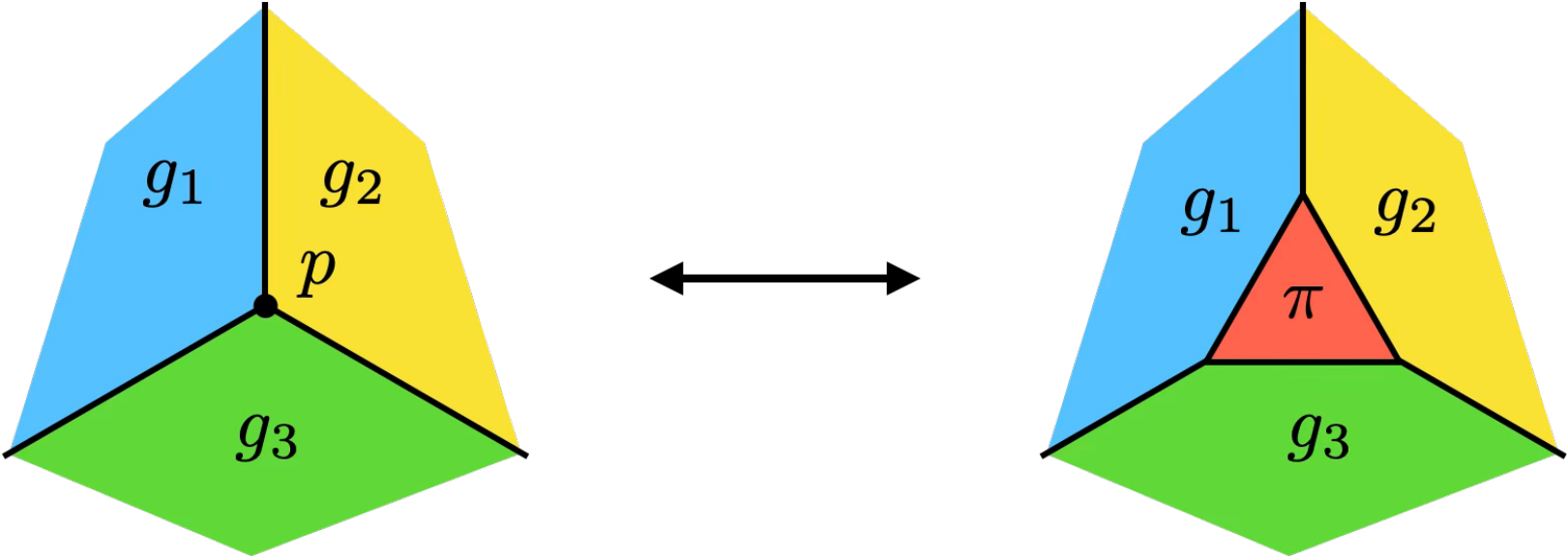}
    \caption{In the case where there exists a nontrivial trivalent junction $p$ for the minimal multiway cut, we consider two different configurations where they differ only near the vicinity of $p$: (left) The replica-symmetric saddle given by the minimal multiway cut. (right) The RSB saddle obtained by replacing a small triangular patch of elements near the junction to $\pi$.}
    \label{fig:RSB}
\end{figure}

To see how this works, consider first the tripartite case $q=3$. Suppose that we are given a minimal cut configuration $g$ of $\mathcal{M}$ such that the three domains meet at a Y junction for least one point $p$ in the bulk.
The interface between different domains must intersect at an equal angle of $2\pi/3$, see Figure \ref{fig:RSB} (left).\footnote{To see why, note that the multiway problem can be thought equally as the problem of finding the lowest energy condition of a tensioned ``soap film'' supported on wireframe defined on the boundary partitions. The equal-angle condition then follows from the force balance condition of the film at point $p$.}
Now zoom in around $p$. Locally, one can always approximate the metric in the neighborhood of $p$ by the flat metric. We consider an alternative configuration of $g$ where it agrees with the minimal cut configuration except for the region near $p$, where it contains a small triangular patch of the RSB element $\pi$, as shown in Figure \ref{fig:RSB} (right). The shape of the triangle depends on whether $n$ is even or odd which we discuss separately below.

Suppose that $n$ is odd. Then $d(\pi,g_i)=(n^2-1)/2$ for $i=1,2,3$. The $\pi$ patch is an equilateral triangle because $\pi$ is equidistant to all the terminal permutations. Suppose that the $\pi$ triangle has side lengths of $1$, the local contribution to the free energy for the Y junction and the $\pi$ solution can be worked out using simple trigonometry:
\begin{equation}
    F^{(3)}_Y = \sqrt{3}(n^2-n), \quad F^{(3)}_\pi = \frac{3}{2}(n^2-1).
\end{equation}
One can easily check that $F^{(3)}_Y>F^{(3)}_\pi$ for any odd integer $n\ge 7$.

Now let $n$ be even. Then $d(\pi,g_i)=n^2/2$ for $i=1,2$ and $d(\pi,g_3)=d(\pi,\id)=(n^2-4)/2$. Since the distances are not equal in this case, the shape of the $\pi$ patch is an isosceles triangle instead. Suppose that the base of the $\pi$ triangle has length $1$ and the legs have length $a$. Then
\begin{equation}
    F^{(3)}_Y(a) = \left(\sqrt{a^2-\frac{1}{4}}+\frac{\sqrt{3}}{2}\right)(n^2-n),\quad F^{(3)}_\pi(a)=\frac{1}{2}\left((1+2a)n^2-4\right).
\end{equation}
Maximizing the difference over $a$ one finds that $a_{\max} = \frac{n}{2\sqrt{2n-1}}$ and
\begin{equation}
    F^{(3)}_Y = \left(\frac{n-1}{2}\sqrt{\frac{1}{2n-1}}+\frac{\sqrt{3}}{2}\right)(n^2-n),\quad F^{(3)}_\pi=\frac{1}{2}\left(\left(1+\frac{n}{\sqrt{2n-1}}\right)n^2-4\right).
\end{equation}
Similarly, one can check that $F^{(3)}_Y>F^{(3)}_\pi$ for any even integer $n\ge 6$. Therefore we conclude that for $q=3$, the RSB saddle dominates for all integer $n\ge 6$.

\begin{figure}[th]
  \centering
  \includegraphics[width=.6\linewidth]{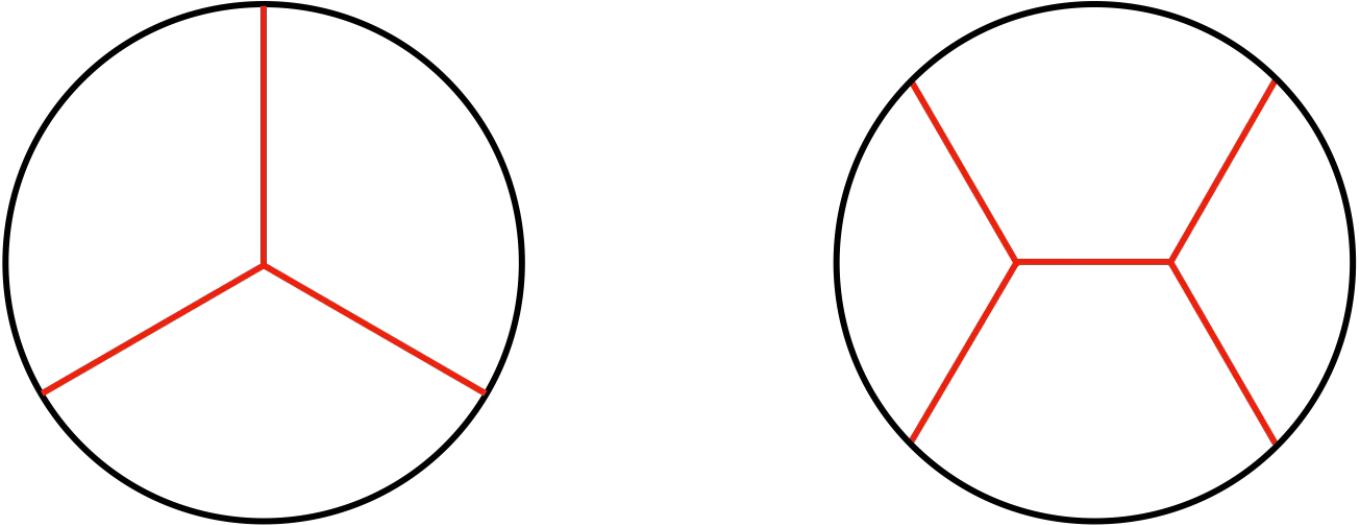}
  \caption{Examples of minimal multiway cuts for three and four boundary regions on a two-dimensional manifold. Note that the intersections of the cuts are always equal-angular trivalent vertices. Configurations with higher-valent vertices are never minimal.}
  \label{fig:Steiner}
\end{figure}

We now deal with the higher-partite cases $q>3$. Consider a configuration $g$ associated to a multiway cut of $\mathcal{M}$. The minimal multiway cut for higher partite settings are given in terms of \emph{Steiner trees} \cite{5a25f51b-2e25-3c6f-9c0a-0a63d8e9f411}, networks of geodesics that meet at equi-angular trivalent vertices, see Figure \ref{fig:Steiner}.
Consider any such trivalent vertex $p$ in the tree (provided that $p$ exists).
As before, we replace a small patch of permutations near $p$ by $\pi$ and denote the three neighboring permutations $g_1,g_2$ and $g_3$. Without loss of generality we can assume $g_3=\id$.\footnote{This is equivalent to translating $\pi'\to \pi' g_3^{-1}$ in the computation of the free energy.} Consider the permutation element $\pi'$ where
\begin{equation}
    \pi'(x_1,x_2,x_3,\cdots,x_{q-1})=(-x_1,-x_2,x_3,\cdots,x_{q-1}).
\end{equation}
In other words, $\pi'$ only reflects the first two coordinates of $x$ and keeps all the others invariant. Then
\begin{equation}
    d(\pi',g_i)=n^{q-3}\frac{n^2-1}{2}, \quad{i=1,2,3}
\end{equation}
for odd $n$, and
\begin{align}
    d(\pi',g_{1,2})=\frac{n^{q-1}}{2},
    \quad d(\pi',g_3)=d(\pi',\id)=n^{q-3}\frac{n^2-4}{2}
\end{align}
for even $n$. Note that $\pi'$ gives lower energy than the full reflector $\pi$. Repeating the same analysis as $q=3$ case, one finds
\begin{align}
    F^{(q)}_Y = n^{q-3}F^{(3)}_Y, \quad F^{(q)}_{\pi'} = n^{q-3}F^{(3)}_\pi, \quad q>3.
\end{align}
One immediately sees that $F^{(q)}_Y>F^{(q)}_{\pi'}$ for integer $n\ge 6$.
Therefore we arrive at the following conclusion:
\begin{proposition}[informal]
  Let $\mathcal{M}$ be two-dimensional Riemannian manifold and $\psi$ be an RTN state associated to a (dense enough) isometric tiling on $\mathcal{M}$.
  If $\partial \mathcal{M} = R_1 \sqcup \cdots \sqcup R_q$ is a $q$-partite ($q>2$) partition of the boundary such that the minimal multiway cut has at least one trivalent junction, then the R\'enyi multi-entropy $S^{(q)}_n(R_1:\cdots:R_q)_\psi$ does \emph{not} equal to the area of the minimal multiway cut for any integer $n\ge 6$.
\end{proposition}

In the case where the underlying metric space is hyperbolic, the size of the $\pi$ patch is finite and is controlled by the local curvature around point $p$. The appearance of the element $\pi$ is thus local in nature, whose effect can be thought of as replacing the plain vertices in the minimal geodesic network by corresponding ``dressed''  vertices, each contributing an negative energy that depends on the local curvature.

We emphasis that the our results in this subsection only guarantees the existence for a RSB saddle for $n\ge 6$. For $2<n< 6$, our construction does not give a more favorable solution, but this does not exclude the possibility for such saddles. For example, if Conjecture \ref{conjecture:higher n} happens to be false and there exists a lower energy minimizer (say $\sigma$) for the single random tensor, then one can replace $\pi'\to \sigma$ in our construction and we expect the bound can be further lowered. We also expect the existence of ``global '' RSB solutions not considered here. The existence of such saddles depends on the underlying geometry and requires an individual analysis for each separate case.

\section{Discussion}
\label{sec:discussion}

In this paper we started with the simple goal of understanding multi-entropy in random tensor networks. What we have found is nonetheless surprising: the nature of the solution is drastically different depending on the R\'enyi index $n$: For $n=2$ the multi-entropy is given by the minimum multiway cut. For $n>2$ this is not true --- the existence of the special element $\pi$ spoils the replica symmetry and renders the multiway cut saddle subdominant. Below we comment on several aspects related to our work.

\subsection*{The specialness of $n=2$}
The key property that makes our proof work for the $n=2$ case is that the terminal permutations are all products of two-element transpositions. This allows the one to simultaneously saturate the estimates given by the number of moved points (Lemma \ref{lem:m(g)}) and the incompatibility condition (i.e. $g(x)=g_i(x)$ for at most one $i$).
For higher $n$ it is not possible to simultaneously saturate both estimates, but both our Conjecture \ref{conjecture:higher n} and numerics seem to suggest that the saturation of Lemma \ref{lem:m(g)} is more important.

It is also interesting to compare our results to the full gravitational counterpart \cite{Gadde:2024taa}: In both RTN and gravity calculations, we see explicit breakdown of replica symmetry for higher R\'enyi indices. In both cases $n=2$ seems to be special in the sense that it is the only R\'enyi index that admits full replica symmetric saddles (in the multipartite case $q \geq 3$). In RTN this saddle allows us to directly recast the $n=2$ R\'enyi multi-entropy in terms of the minimal multiway cut area. In gravity the situation is more complicated, since the bulk solution involves cosmic branes \cite{Dong:2016fnf,Dong:2023bfy}, and one must account for the backreaction to the background geometry.
However it seems that this is as close as one can get to a geometric dual in the same flavor of the RT formula.

Our result provides concrete evidence for the existence of genuine multi-partite entanglement for (holographic) RTN states. In Refs. \cite{Iizuka:2025bcc,Iizuka:2025ioc,Iizuka:2025caq}, it was argued that most holographic states necessarily contain large ($O(1/G_N)$) amounts of genuine multi-partite entanglement. The argument typically involves a bulk evaluation of the $q$-partite \emph{genuine multi-entropy} \cite{Iizuka:2025ioc}, a family of probes built from linear combinations of (R\'enyi) multi-entropies that are able to detect the existence of genuine multi-partite entanglement. The evaluation features extensive use of the multiway cut conjecture for the multi-entropy. In gravity, this requires the questionable $n\to1$ analytic continuation. One can instead work with $n=2$, but the calculation becomes difficult because of the back-reaction from the cosmic branes. However for RTN states, as the multiway cut conjecture has been proven for $n=2$, the argument simply goes through and it is possible to rigorously demonstrate that RTN states possess large amounts of genuine multipartite entanglement.

\subsection*{Other entanglement measures with multiway cut property}

The multi-entropy quantity is formulated as a multi-invariant of a tensor defined via the terminal permutations $g_1, \ldots, g_q$ from Definition~\ref{def:mutli_ent}. The permutations $g_i$ are equidistant and satisfy a special property where for every $i\ne j$, the relative permutation $g^{-1}_ig_j$ is a \emph{fixed-point-free involution} on $X$.
In fact, it is possible to work with different families of twist operators $\sigma_1,\cdots,\sigma_p$. As long as the fixed-point-free property is still satisfied by the new family\footnote{One can verify that the equidistant condition follows if one requires $\sigma_i^{-1}\sigma_j$ to be a fixed-point-free involution for all pairs $i\ne j$.}, the conclusion of our main result, namely Theorem~\ref{thm:RTN}, will still be valid. Stated more precisely:
\begin{proposition}
  Let $p> 2$ be an integer and $\Sigma=(\sigma_1,\cdots,\sigma_p): \sigma_i\in \Sym(X)$ be a family of permutations. Consider the problem of determining the multi-invariant $\mathcal{E}(\sigma_1,\cdots,\sigma_p)$ on an RTN state defined on a graph $G=(V,E)$ (with weight $w:V\to \mathbb{R}_+$) with boundary vertices $R_i,\cdots,R_p\subset V$:
  \begin{align}
    &\min_{g:V\to \Sym(X)} F_\Sigma(g), \quad  F_\Sigma(g):=\sum_{e=(u,v)\in E} w(e)d(g(u),g(v)),
  \end{align}
  subject to
  \begin{equation}
    g(v) = \sigma_i,\quad v\in R_i.
  \end{equation}
  If for every pair $i\ne j$, the relative permutation $\sigma_i^{-1}\sigma_j$ is a \emph{fixed-point-free involution}, i.e.
  \begin{equation}
    \sigma_i^{-1}\sigma_j(x) \ne x ~\forall x\in X\quad \text{and} \quad (\sigma_i^{-1}\sigma_j)^2 = \id.
  \end{equation}
  Then the minimal value of $F_\Sigma$ is given by the multiway cut across $G$:
  \begin{align}
    \min_g F_\Sigma(g) = \frac{|X|}{2} \mathcal{A}(R_1:\cdots:R_p)
  \end{align}
  Moreover, a minimizer is obtained by taking any minimum multiway partition $V=\Gamma_1\sqcup\cdots \sqcup \Gamma_p$ and setting
  \begin{equation}
    g(v)=\sigma_i, \quad v\in \Gamma_i.
  \end{equation}
\end{proposition}
The proof of the proposition is left as an exercise for the reader.

As an illustration, one can consider the whole family of translations on $X = \mathbb Z_2^{q-1}$: for all $z \in \mathbb Z_2^{q-1}$, define the permutation $\sigma_z(x) = x + z$. This family of $2^{q-1}$ permutations satisfies the property above: for all $z \neq w$, $\sigma_z^{-1}\sigma_w = \sigma_{z+w}$, which is a fixed-point-free involution on $X$. The multi-entropy case corresponds to a subfamily of size $q$.

Lastly, we stress that the equidistant condition alone is not sufficient to guarantee the result. For example, the twist operators for the $n=3$ R\'enyi negativity
\begin{equation}
  X=\{1,2,3\},\quad g_1 = (123),\quad g_2 = (321),\quad g_3 = \id
\end{equation}
are equidistant since $d(g_1,g_2)=d(g_2,g_3)=d(g_3,g_1) = 2$. However, the relative permutations $g_i^{-1}g_j$ are all 3-cycles so the involutive condition fails. Indeed, the value of the R\'enyi negativity is not given by a multiway cut, see e.g. \eqref{eq:negativity_sol} and Figure \ref{fig:TN_cuts2}.

\section*{Acknowledgments}
We thank Yuan-Tang Chou for help on the numerical analysis. S.L.~would like to thank Universit\'e de Toulouse, CERN and National Tsing Hua University for hospitality where part of this work was completed.
M.H.~and I.N.~were supported by the ANR project \href{https://www.ceremade.dauphine.fr/dokuwiki/anr-tagada:start}{TAGADA} grant number ANR-25-CE40-5672.

\small
\bibliographystyle{ourbst}
 \bibliography{post.bib}

\end{document}